\documentclass[a4paper,USenglish, 11pt]{article}

\usepackage{graphicx}
\usepackage{latexsym,amssymb,amsmath}
\usepackage{tkz-graph}

\usepackage{multirow}
\usepackage{fullpage}

\usepackage{authblk}

\newcommand{\itemshort}{
\setlength{\itemsep}{0mm}
\setlength{\parskip}{0mm}
}
\newtheorem{observation}{Observation}

\newtheorem{definition}{Definition}
\newtheorem{lemma}{Lemma}
\newtheorem{theorem}{Theorem}
\newtheorem{corollary}{Corollary}
\newenvironment{proof}{\noindent {\bf Proof: }}{\hfill$\square$}
\newcommand{\qed}{\hfill$\square$}

\title{A self- stabilizing algorithm for maximal matching in link-register model in $O(n\Delta^3)$ moves}
\date{}

\author[1]{Johanne Cohen}
\author[1]{Georges Manoussakis\thanks{This work was partially funding by DIGITEO}}
\author[2]{Laurence Pilard}
\author[2]{Devan Sohier}
\affil[1]{LRI, CNRS, Universit\'e Paris Sud, Universit\'e Paris-Saclay, France\\  \texttt{\{~johanne.cohen,georges.manoussakis\}@lri.fr}}

\affil[2]{LI-PaRAD, Universit\'e Versailles-St. Quentin, Universit\'e Paris-Saclay, France\\
  \texttt{\{~laurence.pilard, devan.sohier\}@uvsq.fr}}

\newcommand{\DEL}[1]{}

\newcommand{\Exec}{\mathcal{\mathcal{E}}}
\newcommand\Voisin{N}

\newcommand{{\minN}}{s} 
\newcommand{{\maxN}}{t} 

\newcommand{\PRaba}{PRabandonment}
\newcommand{\PRReset}{PRreset}
\newcommand{\FCTpointingto}{Correct\_register\_value}
\newcommand{\Idle}{\emph{Idle}}%
\newcommand{\You}{\emph{You}}%
\newcommand{\Other}{\emph{Other}}%

\newcommand{\stylegraphe}{
	\tikzstyle{ann} = [circle,draw=black]
 
	\tikzstyle{node}=[draw,circle,thick]
	\tikzstyle{vertex carre}=[draw,rectangle,line width=2pt]
	\tikzstyle{vertex circle}=[draw,circle,minimum size=0.6cm,inner sep=0pt]
	\tikzstyle{matched}=[line width=2pt]
	\tikzstyle{matched node} = [draw,circle,minimum size=0.6cm,line width=2pt,-,black]
	\tikzstyle{pointer} = [draw, thick,->,>=stealth,black,shorten >=1pt]

	\tikzstyle{alternating} = [circle, scale =0.5,draw=black]
	\tikzstyle{couleur1} = [fill=white]
	\tikzstyle{couleur2} = [fill=gray!20]
	\tikzstyle{couleur3} = [fill=black!80]
	\tikzstyle{couleur4} = [fill=black!30]
	\tikzstyle{couleur5} = [fill=gray!60]
	\tikzstyle{couleur8} = [fill=gray!80]
	 
	\tikzstyle{matched edge} = [draw,line width=3pt,-,black]
	\tikzstyle{edge} = [draw,-,black]
	\tikzstyle{legend}=[rectangle, thin,minimum width=2.5cm, minimum height=1.2cm,black]
}

\begin{document}
\maketitle

\begin{abstract}
 In the matching problem, each node maintains a pointer to one of its neighbor or to $null$, and a maximal matching is computed when each node points either to a neighbor that itself points to it (they are then called married), or to $null$, in which case no neighbor can also point to $null$.  This paper presents a self-stabilizing distributed algorithm to compute a maximal matching in the link-register model under read/write atomicity, with complexity {$O(n\Delta^3)$} moves under the adversarial distributed daemon, where $\Delta$ is the maximum degree of the graph. \end{abstract}


\section{Introduction}

The matching problem consists in building disjoint pairs of adjacent nodes. The matching is maximal if no new pair can be built, \emph{i.e.}, if among any two adjacent nodes, at least one of them is part of a pair. This problem has a wide range of applications in networking and parallel computing, such as the implementation of load balancing.

All nodes participating in the distributed matching algorithm maintain a variable, called the pointer, that can take any neighbor as value, or the special value $null$; a node pointing to $null$ is single, and two nodes pointing one toward the other are called married. The set of pairs of married nodes constitutes the matching.

The matching problem definition is thus local; however, the building process needs to take long range phenomena into account, to avoid cycles in which node pointers form a cycle of size greater than 2. To break such symmetries, we suppose that all nodes have a unique identifier.

\section{State of the art}

The matching problem has recently received much attention, both in graph theory and in distributed computing.

For instance, in graph theory, some (almost) linear time approximation sequential algorithm for the maximum weighted matching problem (the maximum matching being the largest matching in terms of the weight of the edges) have  recently been studied \cite{DrakeH03,Preis1999}. 

Self-stabilizing algorithms for computing maximal matching have been designed  in various models (anonymous  network~\cite{Asada2015}  or not \cite{TurauH11a}, weighted or unweighted, see \cite{GuellatiK10} for a survey).  For an unweighted graph, Hsu and Huang \cite{HsuH92} gave the first  self-stabilizing algorithm and proved a bound of $O(n^3)$ on the number of moves under a sequential adversarial daemon. Hedetniemi \emph{et al.} \cite{HedetniemiJS01} completed the complexity analysis proving a $O(m)$ move complexity. Manne \emph{et al.} \cite{ManneMPT09} gave  a self-stabilizing algorithm that converges in $O(m)$ moves under a distributed adversarial daemon. Cohen \emph{et al.} \cite{cohen2016self} extend this result and  propose a randomized self-stabilizing algorithm for computing a maximal matching in an anonymous network. The complexity is $O(n^{2})$ moves with high probability, under the adversarial distributed daemon.

All these algorithms work in a state model, in which nodes can directly access the variables of adjacent nodes. This model was introduced by Dijkstra in \cite{Dijk74}. However, this model fails to capture the aynchronous phenomena that happen in many real-life distributed systems, which led \cite{DIM93} to introduce the link-register model. In this model, communications are abstracted by registers in which nodes can write and read values. Atomicity conditions define the granularity of the algorithm. A variety of atomicity conditions exist (see \cite{JohnenH08} for a survey and some results on the strength of the different atomicities).

Read/write atomic registers are registers associated to each (directed) link, in which one of the nodes can write, and the other read; each read or write operation on the register is atomic (meaning that it cannot be interrupted), but a read in the register can happen arbitrarily long after the previous write, forcing the reading process to act based on outdated values (as opposed to what happens in Dijkstra-type state model). Read/write atomic registers can be implemented over message-passing models (at a large cost however), as in \cite{MostefaouiPRJ17} for instance. This kind of atomicity is the strongest, meaning that algorithm written under this model also solve the problem under the other classical models.

The possible occurence of faults in the execution of the algorithm is taken into account with the paradigm of self-stabilisation, as defined in \cite{Dijk74}. A fault (or a sequence of faults) can lead the system to an arbitrary configuration of the processes and registers, starting from which the execution (seen as a completely new execution starting from this arbitrary configuration) must eventually resume a correct behavior in finite time. A complete study of this notion can be found in \cite{Dole00}.

We propose the first distributed self-stabilizing algorithm solving the matching problem in a link-register model with read/write atomicity. The algorithm is presented under the form of guarded rules (usual for state model algorithms, but as far as we are aware of, never used before for link-register algorithms). This allows to underscore the granularity of the model, each configuration being the result of the application of an arbitrary subset of guarded rules to the previous configuration.

The algorithm works as follows: the lowest id node of a pair proposes to its neighbor, that accept (or not) the proposition; then the marriage is confirmed in three steps. This scheme can be interrupted at any point, either because of another marriage being concluded, or because of a faulty initialization. Once the marriage is reached, the married nodes do not take any more move, and the algorithm is eventually silent. At worst, the algorithm has to take $O(n\Delta^3)$ moves before reaching a maximal matching.

\section{Model}

The system consists in a set of processors $V$ and a set $E\subset V\times V$ of links. We suppose that communications are bidirectional, so that $(u, v)\in E\Rightarrow(v, u)\in E$. Considering two adjacent processors $u$ and $v$ (\emph{i.e.}, $(u, v)\in E$), there exists a register $r_{uv}$ in which $u$ is the only process allowed to write, and that $v$ can read.

The set of neighbors of a process $u$ is denoted by $\Voisin(u)$ and is the set of all processes adjacent to $u$, and $\Delta$ is the maximum degree of $G$.

All nodes have the same variables; if $var$ is a variable, $var_u$ denotes the instance of this variable on process $u$. Each node $u$ has a unique id $id_u$; in the following, for the sake of simplicity, we do not distinguish between $u$ and $id_u$.

In the matching problem, each node $u$ maintains a variable $p_u\in\Voisin(u)\cup\{null\}$ indicating the neighbor it is married or attempting to marry.

A configuration describes the situation of the algorithm at a point in its execution: in the link-register model, it is the vector of the values of all variables and registers. In particular, a configuration solves the maximal matching if it is such that $\forall u, (p_u\neq null \Rightarrow p_{p_u}=u)\land (p_u=null\Rightarrow \forall v\in\Voisin(u), p_v\neq null)$. The first part of this specification means that if a node points to one of its neighbor, this neighbor points to it; the second one implies maximality: if a node points to no other node, none of its neighbors is in the same situation, since they could marry and create a larger matching.

The algorithm is presented under the form of a set of guarded rules. A guarded rule consists in a guard which is a predicate on the values of the variables of a node, and an atomic action that can be executed if the guard is true. To respect read/write atomicity, if a guard refers to the value of a neighbor's register (which implies the reading of this register), the associated action cannot write in a register. In particular, we decided to introduce the $Write$ guarded rule, that writes the adequate value in a register; other actions never write in any register. Thus, all actions consist either in readings in neighbor's registers and taking local actions, or in writing in its own register, which respects the read/write atomicity. A guarded rule is activable in a configuration if its guard is true.

Link-register algorithms are generally presented as an infinite loop of readings, local actions and writings for each node, and the chosen atomicity allows to decide the points in the algorithm at which the execution of a node can be suspended to let another node take over. We chose to show more explicitly the points at which a node suspends action, and thus the result in terms of configuration of each of its atomic actions: the next configuration in an execution is obtained by applying one or several actions of guarded rules whose guards are true.

A \emph{daemon} is a predicate on the executions. We consider only the most powerful one: the \emph{adversarial distributed daemon} that, at each step, picks a nonempty subset of activable rules (at most one for each node) and executes them. The next configuration is then the result of the concurrent application of the actions of all these rules (as only one rule can be selected for each node, and that no action can modify the variables of other nodes, this concurrent application yields unambiguous result).

The moment when a process $u$ writes in register $r_{uv}$ is the time starting from which the written value $r_{uv}$ is available to $v$, thus, the writing is analogous to a message reception by $v$ in a message-passing model. A node $u$ reads in its register $r_{uv}$ in all guards. This allows it to check that the writing register of a node has reached its correct value. This can be paralleled with an acknowledgement.

An \emph{execution} is then an alternate sequence of configurations and actions  $ \Exec = (C_0,A_0,\ldots,C_i,A_i, \ldots)$, such that:\begin{itemize}
\item $C_i$ is a configuration;
\item $A_i$ is a nonempty set of couples $(R, u)$ such that $R$ is activable for node $u$ in configuration $C_i$, and no two of these couples concern the same node;
\item $C_{i+1}$ is obtained from $C_i$ by executing all actions in $A_i$.
\end{itemize}
We then write $C_i\mapsto_\Exec C_{i+1}$, or $C_i\mapsto C_{i+1}$ if the execution is clear from the context, $C_i\mapsto^*C_j$ if $i\leq j$, and $C_i\mapsto^+C_j$ if $i<j$; in this case, we say for any action $(R, u)\in A_k$ with $i\leq k\leq j$ that there has been a transition $(R, u)$ between $C_i$ and $C_j$.

An algorithm is \emph{self-stabilizing} for a given specification, if there exists a sub-set $\cal L$ of configurations such that every execution starting from a configuration in $\cal L$ verifies the specification  (\emph{correctness}), and every execution, starting from any configuration, eventually reaches a configuration in $\cal L$ (\emph{convergence}). A configuration is \emph{stable} if no process is activable in the configuration. The algorithm presented here, is \emph{silent}: its maximal executions are finite and end in a stable configuration. We call legitimate a stable configuration verifying the maximal matching specification.

\section{Algorithm}

The presented algorithm is based on the algorithm by Manne \emph{et al}~\cite{ManneMPT09} written under the state model. A marriage is contracted in two phases: (1) \emph{the selection} of the edge to add to the matching and (2) \emph{the confirmation} or the \emph{lock} of the edge in the matching in three steps. 


\paragraph{Variables description:}
Each node $u$ has two local variables. Variable $p_u$ is the identifier of the node $u$ points to: nodes $u$ and $v$ are said to \emph{be married} to each other if and only if $u$ and $v$ are neighbors, $p_{u}$ points to $v$, and  $p_{u}$ points to $u$.  We also use a variable $m_u$ indicating  the progress of $u$'s marriage.
 $m_u\in\{0,1,2\}$. 

Each node $u$ has a four bit register $r_{uv}$ for each of its neighbors $v$. The first two bits $r_{uv}.p$ can take the value Idle if $u$ points to $null$ (ie $p_u=null$), You if it points to $v$, and Other if it points to a node $\neq v$. The last two bits $r_{uv}.m$ can be 0, 1 or 2, indicating the progress of u's marriage.



\paragraph{Algorithm description:}
The $Seduction$ and $Marriage$ rules implement the selection of an edge for the matching: they set the $p$ variable of a node to a candidate for a marriage. First, the node with the smallest id in a pair executes the $Seduction$ rule, to which the node with the highest id can respond by executing the $Marriage$ rule. This asymetrical process avoids situations with nodes trying to seduce neighbors in a cycle, that could be reproduced by an adversarial daemon.

Under read/write atomicity, an offset is possible between the value of the local variables of a node ($p, m$) and the value of its registers. In order to avoid infinite executions during which the distributed daemon lets a node $u$ attempt to marry a neighbor $v$ just at the same time when $v$ abandons its attempt to marry $u$, and then conversely at the next step, it is necessary to design a mechanism locking progressively a marriage. We achieve that with variable $m$, which takes values in $\{0,1,2\}$: except for faulty initialization, $m_u=0$ means that $u$ did not start locking any marriage, and $m_u\geq 1$ means that $u$ has a neighbor $v$ such that $p_u=v\land p_v=u$. If $m_u=1$ then the marriage lock is in progress and if $m_u=2$ then the lock is done. $m$ is incremented in the execution of rule $Increase$.


The $Reset$ rule ensures that local variables $p$ and $m$ for a node $u$ have consistent values. It is executed when predicate $\PRaba(u)$ or \PRReset(u) is true. $\PRaba(u)$ means that $u$'s marriage process should be restarted if $u$ is trying to marry a node $v<u$ that is not seducing it, or if $u$ is trying to seduce a node that is already married (at least, when the registers of $v$ indicate that). This last case can happen when $v$ is responding to several proposals at the same time, while the first one is provoked by ``bad'' initializations. $\PRReset(u)$ indicates a discrepancy between the steps taken in the locking mechanism by the two processes involved in it. This is due to ``bad'' initializations.\\~\\

\paragraph{Predicates and functions of the algorithm:}

{\small 
\begin{tabbing}

$\FCTpointingto(u,a)\equiv$ \= if $p_u=null$ then return $(\Idle,0)$\\
\> else if $p_u=a$ then return $(\You,m_u)$\\
\> else return $(\Other,m_u)$\\[3pt]


aaa\=aaaaaaaaiaaaaa\=ax\=aaaaaaa\=aa\=aaaaaii\=aa\=aa\=\hspace*{10cm}\=\kill
$\PRaba(u)\equiv [p_u\neq null\land\,(r_{p_{_u}\!u}.p \neq You \land (u>p_u \lor m_{u}\neq 0 ))$\\
\>\>\>\> $\lor (r_{p_{_u}\!u}=(Other,2) \land u<p_u)\,]$\\[3pt]

$\PRReset(u)\equiv (p_u\neq null)\land(r_{p_{_u}\!u}.p=You) \land  ($\\
\>\>\> $(m_{u}=0 \land r_{p_{_u}\!u}.m =2)$\\
\>\>$\lor$\>$(m_{u}=2 \land  r_{p_{_u}\!u}.m =0)$\\
\>\>$\lor$\>$(m_{u}=0 \land  r_{p_{_u}\!u}.m =1 \land u>p_u)$\\
\>\>$\lor$\>$(m_{u}=1 \land  r_{p_{_u}\!u}.m =0 \land u<p_u)$\\
\>\>$\lor$\>$(m_{u}=1 \land  r_{p_{_u}\!u}.m =2 \land u<p_u)$\\
\>\>$\lor$\>$(m_{u}=2 \land  r_{p_{_u}\!u}.m =1 \land u>p_u))$
\end{tabbing}
}

\paragraph{Rules for each node $u$:}

{\small 
\begin{tabbing}

$\mathbf{\forall a\in N(u),}$\\[3pt]

aaaaa\=aaaaaaaaiaaaaa\=ax\=aaaaaaa\= \kill 
\>\textbf{Write(a)}\> \textbf{::}
\> $r_{ua} \neq \FCTpointingto(u,a)$ \\
\>\>\>\> $\to$ $r_{ua}:=\FCTpointingto(u,a)$\\[3pt]

aaaaa\=aaaaaaaaiaaaaa\=ax\=taaaiaaaaa\=\kill 
\>\textbf{Seduction(a)}\> \textbf{::}
\> $p_{u}= null$ \> $\land~r_{ua} = \FCTpointingto(u,a)$\\
\>\>\>\>  $\land~r_{au}=(\text{\Idle}, 0) \land (u<a) ~~~ \to (p_{u}, m_{u}):= (a, 0)$\\[3pt]

\>\textbf{Marriage(a)}\> \textbf{::}
\> $p_{u}= null$ \> $\land~r_{ua} = \FCTpointingto(u,a)$\\
\>\>\>\> $\land~r_{au}=(\text{\You}, 0) \land (u>a) ~~~ \to (p_{u}, m_{u}):= (a, 0)$\\[3pt]

\=aaaaaaaaa\=ax\=xxx\=aa\=aa\=aa\=aaaaaii\=aa\=aa\=\hspace*{10cm}\=\kill
\>\textbf{Increase}\> \textbf{::}
\> $p_u\neq null \land r_{up_{_u}} = \FCTpointingto(u,p_u)$\\
\>\>\>\>  $\land$ \> $ (r_{p_{_u}\!u}.p=You) \land  ($\\
\>\>\>\>\>\>\> $(m_{u}=0) \land  [\,(u<p_u \land r_{p_{_u}\!u}.m =1) \lor (u>p_u \land r_{p_{_u}\!u}.m=0)\,]$\\
\>\>\>\>\>\> $\lor$ \> $(m_{u}=1) \land [\,(u<p_u \land  r_{p_{_u}\!u}.m =1) \lor (u>p_u \land  r_{p_{_u}\!u}.m =2)\,]$\\
\>\>\>\>\> $) \to$\\
\>\>\>\>\>\>\>\> $m_{u}:=  m_{u} +1 $\\[3pt]

\=aaaaaa\=ax\=taaaaaaaai\=aaaaaaaaaaaaaaaaaaaaaaaaaaaaaaaaaaaa\=\kill
\>\textbf{Reset}\> \textbf{::}
\> $p_u\neq null$ \> $\land~r_{up_{_u}} = \FCTpointingto(u,p_u)$\\
\>\>\>\> $\land~(\PRaba(u) \lor \PRReset(u))$ \> $\to (p_{u}, m_{u}):=(null, 0)$

\end{tabbing}
}

\paragraph{About the rules:}


A node $u$ has $3\deg(u)+2$ rules: one rule $Write$, $Seduction$ and $Marriage$ for each neighbor, plus one rule $Increase$ and one $Reset$, that are not associated to a link. Given a neighbor $v$ of $u$, $u$ can be activable only for rule $Marriage(v)$ or $Seduction(v)$, because the first one necessitates that $u>v$ and the second that $u<v$. 


\begin{definition}[$v$-$Increase$/$v$-$Reset$, $v$-rule and $R(-)$ rule] 
Let $u$ and $v$ be two nodes. 
We say that node $u$ is activable in the configuration $C$ for a $v$-$Increase$ (resp. $v$-Reset) rule if, in $C$, $u$ is activable for an $Increase$ (resp. $Reset$) rule and $p_u=v$.
We say that node $u$ is eligible in the configuration $C$ for a $v$-rule if $u$ is eligible for one of the following rule: $\{Write(v), Seduction(v)$, $Marriage(v), v$-$Increase, v$-$Reset\}$ in $C$.
Finally, let $R$ be any rule among $Write$, $Marriage$ and $Seduction$. We say that $u$ is eligible for a $R(-)$ rule, if there exists a neighbor of $u$, say $a$, such that $u$ is eligible for a $R(a)$ rule. 
\end{definition}



\begin{observation}\label{lem:LP:Exclusive2}
Let $u$ be a node and $C$ be a configuration. In $C$, we have:
\begin{enumerate}
\item if $p_u=null$ then: 
\begin{itemize}
\itemshort
\item $u$ is not eligible for $Increase$ nor $Reset$;
\item $\forall v\in \Voisin(u): u$ is eligible for at most one rule among the set of rules $\{Write(v)$, $Marriage(v)$, $Seduction(v)\}$;
\end{itemize}
\item\label{lem:LP:notnull} if $p_u\neq null$ then:
\begin{itemize}
\itemshort
\item $u$ is not eligible for $Marriage(-)$ nor $Seduction(-)$; 
\item $u$ is eligible for at most one rule among the set of rules $\{Write(p_u)$, $Increase$, $Reset\}$; moreover, if this rule is an $Increase$ (resp. $Reset$), this is necessarily a $p_u$-$Increase$ (resp. $p_u$-$Reset$);
\item and $\forall x\in \Voisin(u)\setminus \{p_u\}:$  among all the $x$-rules, $u$ can only be eligible for $Write(x)$;
\end{itemize}
\end{enumerate}
\end{observation}

\paragraph{Execution examples:}
Below is an execution of the algorithm under the adversarial distributed daemon. Figure~\ref{fig:execution:exemple}a shows the initial state of the execution. Node identifiers are indicated inside the circles. Black arrows show the content of the local variable $p$ and the absence of arrow means that $p=null$. $s{:}Write(t)$ means that node $s$ executes the $Write(t)$ rule. 

Consider an initial configuration (Figure~\ref{fig:execution:exemple}a in which variable and register values as follows: $(p_s,m_{s})=(p_t,m_{t})=(null,0)$, $r_{st}=(You,2)$ and $r_{ts}=(Idle,0)$. Thus, initially, nodes $s$ and $t$ are not matched. Since the local variables of $s$ are not consistent with register $r_{st}$, node $s$ executes $Write(t)$ in order to set $r_{st}=(Idle,0)$  (Figure~\ref{fig:execution:exemple}b). In this execution, we take $s<t$.

Now, since $s$ and $t$ are not matched they can start a selection process in order to marry. The node with the smallest identifier, $s$, starts the process and thus, since $r_{ts}=(Idle,0)$, $s$ executes a $Seduction(t)$ rule in order to set $p_s$ to $t$. $s$ is then eligible to execute a series of $Write(v)$ rules to update its registers. Consider an execution it executes at some point $Write(t)$ (Figure~\ref{fig:execution:exemple}c).

Once register $r_{st}$ is updated, node $t$ answers to the proposition of $s$ by executing a $Marriage(s)$ rule, setting $p_t=s$. It is then eligible to execute a series of $Write(v)$ rules to update its registers. As long as $t$ does not update its $r_{ts}$ register, the process of locking the marriage cannot start since $s$ needs $r_{ts}.p=You$ in order to start increasing its $m$ variable. So assume $t$ updates its $r_{ts}$ register with a $Write(s)$ (Figure~\ref{fig:execution:exemple}d).

From this point, both nodes point towards each other. 
The locking process of the marriage starts from this point. 
First, the node with the highest identifier, that is $t$, sets its $m$ variable to $1$ with a $s$-$Increase$ and then updates its registers (see Figure~\ref{fig:execution:exemple}e). Then node $s$ executes a $t$-$Increase$ and sets $m_s=1$ followed by a $Write(t)$ rule to update its register  (Figure~\ref{fig:execution:exemple}f). It executes yet another $t$-$Increase$ rule to set $m_s=2$ (Figure~\ref{fig:execution:exemple}g). This execution of two consecutive $t$-$Increase$ by $s$ guarantees that $t$ has correct register values. 
Finally, after registers have been updated accordingly, $t$ executes a last $s$-$Increase$ rule to set $m_t=2$. At this point the matching of $s$ and $t$ is locked.

One might wonder why the $Increase$ rule is not alternately executed by $s$ and $t$. Indeed, $s$ executes two consecutive $t$-$Increase$ to set its $m$ variable from $0$ to $1$ and then to $2$, while $t$ does not change its $m$ value. Actually, the algorithm does not converge if nodes would perform the $Increase$ rule alternately. These two consecutive $Increase$ are a key point on the lock process of a marriage. 


\newcommand{\stylechaine}{
\tikzstyle{legend}=[rectangle, thin,minimum width=2.5cm, minimum height=0.8cm,black]
	\tikzstyle{vertex}=[draw,circle,minimum size=0.3cm,fill=black,inner sep=0pt]
	\tikzstyle{mvertex}=[draw,circle,minimum size=0.3cm,fill=red,inner sep=0pt]
	\tikzstyle{vertex retour}=[draw,rectangle,minimum size=0.25cm,fill=black,inner sep=0pt]
	\tikzstyle{matched edge} = [draw,line width=1pt,-,red]
	\tikzstyle{edge} = [draw,thick,-,black!40]
	\tikzstyle{medge} = [draw,thick,-,red!40]
	\tikzstyle{pointer} = [draw,  thick,->,>=stealth,bend left=40,black,shorten >=1pt]
}
\newcommand{\drawPointeur}[2]{\draw[pointer] (#1) to [bend left=45] (#2);}

\newcommand{\drawTopologie}{
\stylegraphe
\node[node] (v-2) at (2.8,0.5) {$s$};
\node[node] (v-1) at (1.4,0.5) {$t$};
\draw[edge] (v-1)--(v-2);
}

\newcommand{\drawVariableS}[3]{
 \node[] (S) at (2.4,0.6) {\tiny (#1,#2)};
\node[legend]  at (2.6,0.25) {\small $m_s=#3$};}
 
 \newcommand{\drawPointeurS}{\draw[pointer] (v-2.north) to [bend right=25] (2.3,0.7);}
 \newcommand{\drawPointeurT}{\draw[pointer] (v-1.north) to [bend left=25] (1.9,0.7);}

\newcommand{\drawVariableT}[3]{
\node[] (T) at (1.8,0.6) {\tiny (#1,#2)};
\node[legend]  at (1.5,0.25) {\small $m_t=#3$};}

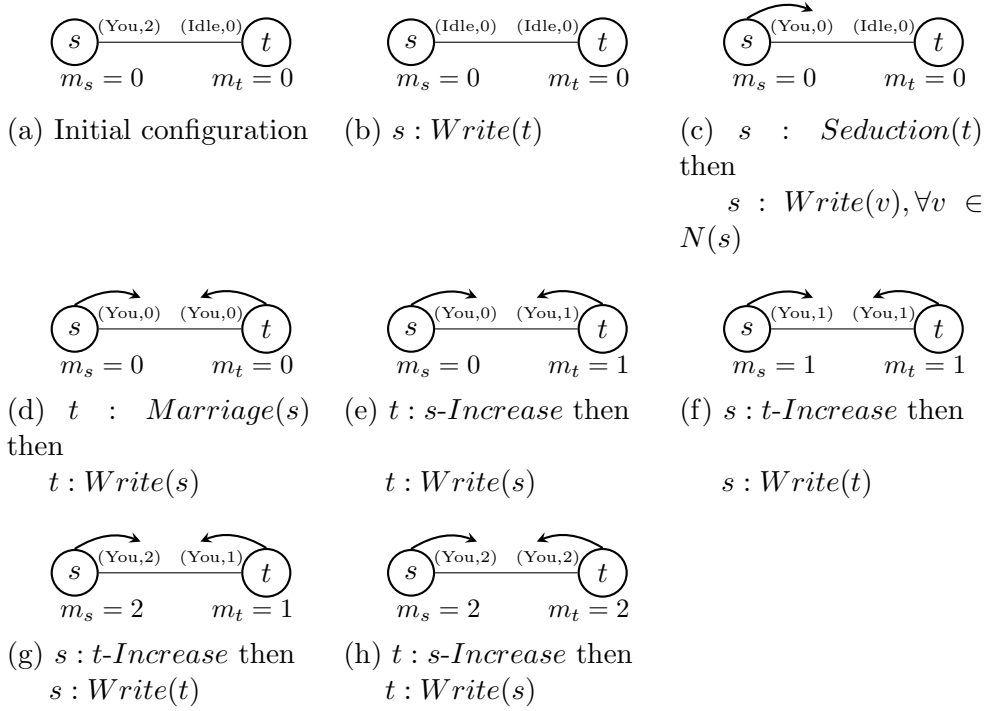
\begin{figure}[h]
\centering
\begin{tabular}{p{4cm}p{4cm}p{4cm}}
 \begin{tikzpicture}[xscale =  -1.8,yscale=2]
\drawTopologie
\drawVariableS{You}{2}{0}
\drawVariableT{Idle}{0}{0}
\end{tikzpicture}

&	
 \begin{tikzpicture}[xscale =  -1.8,yscale=2]
\drawTopologie
\drawVariableS{Idle}{0}{0}
\drawVariableT{Idle}{0}{0}
\end{tikzpicture}
&
\begin{tikzpicture}[xscale =  -1.8,yscale=2]
\drawTopologie
\drawVariableS{You}{0}{0}
\drawVariableT{Idle}{0}{0}
\drawPointeurS
\end{tikzpicture}
\\[-2em]

(a) Initial configuration  &(b)  $s:Write(t)$ & (c)  $s:Seduction(t)$ then \\ 
& & \hspace*{12pt} $s:Write(v), \forall v\in N(s)$ \\[1em]

\begin{tikzpicture}[xscale =  -1.8,yscale=2]
\drawTopologie
\drawVariableS{You}{0}{0}
\drawVariableT{You}{0}{0}
\drawPointeurS  \drawPointeurT
\end{tikzpicture}
&	
 \begin{tikzpicture}[xscale =  -1.8,yscale=2]
\drawTopologie
\drawVariableS{You}{0}{0}
\drawVariableT{You}{1}{1}
\drawPointeurS \drawPointeurT
\end{tikzpicture}
&
 \begin{tikzpicture}[xscale =  -1.8,yscale=2]
\drawTopologie
\drawVariableS{You}{1}{1}
\drawVariableT{You}{1}{1}
\drawPointeurS \drawPointeurT
\end{tikzpicture}
\\[-1em]

(d) $t:Marriage(s)$ then & (e) $t:s$-$Increase$ then  & (f) $s:t$-$Increase$ then\\
\hspace*{12pt}    $t:Write(s)$ &  \hspace*{12pt}   $t:Write(s)$ &  \hspace*{12pt}   $s:Write(t)$ \\[1em]

 \begin{tikzpicture}[xscale =  -1.8,yscale=2]
\drawTopologie
\drawVariableS{You}{2}{2}
\drawVariableT{You}{1}{1}
\drawPointeurS  \drawPointeurT
\end{tikzpicture}
&	
 \begin{tikzpicture}[xscale =  -1.8,yscale=2]
\drawTopologie
\drawVariableS{You}{2}{2}
\drawVariableT{You}{2}{2}
\drawPointeurS \drawPointeurT
\end{tikzpicture}
& \\[-1em]

(g) $s:t$-$Increase$ then & (h)  $t:s$-$Increase$ then\\
\hspace*{12pt} $s:Write(t)$ &  \hspace*{12pt} $t:Write(s)$

	
\end{tabular}
\caption{A typical execution of the algorithm. The absence of arrow means that the $p$-variable is equal to $null$. }
\label{fig:execution:exemple}
\end{figure}

%
%
%
%

\section{Proof}
\subsection{State of an edge}
We first focus on an edge $({\minN},{\maxN})$ of the matching $\mathcal{M} = \{ (a,b) \in E : p_a=b \land p_b=a\}$ built by our algorithm when ${\minN} < {\maxN} $. 
In particular, we focus on values of local variables and registers of this edge in some chosen configurations.  

\begin{definition}\label{def:state}
Let $({\minN},{\maxN})$  be an edge with ${\minN} < {\maxN}$. We say that in a configuration $C$, the edge $({\minN},{\maxN})$ is in state $(You,\alpha,\beta)$ if 
{$(p_{\minN}={\maxN} \land p_{\maxN}={\minN}) \land  (m_{\minN}=\alpha \land m_{\maxN}=\beta)$}
\end{definition}

If an edge $({\minN},{\maxN})$ is in state  $(You,\alpha,\beta)$, then this edge belongs to the matching. Unfortunately, due to some ``bad'' initialization for instance, this edge can be removed from the matching at some point of the execution. In the following, we characterize an edge that belongs to the matching and that will forever remain in it.

A correct state corresponds to the situations appearing in the Figure~\ref{fig:execution:exemple} starting from step (d). Starting from a configuration where edge $(s,t)$ is in a correct state, the two nodes, one after the other, execute  $Increase$ and $Write$ rules. A link is in an updated correct state when all registers of the edge are updated (and so exactly one node among  $s$ and $t$ is eligible for an $Increase$), while it is toUpdate when the register of one of the two nodes is not up to date (and so exactly one node among  $s$ and $t$ is eligible for a $Write$).

\begin{definition}[Updated correct state] \label{def:updatedCS} 
Consider an edge $({\minN},{\maxN})$ with ${\minN} < {\maxN}$ in state $(You,\alpha,\beta)$ in a configuration $C$. This link is in an \emph{updated correct state} if \\
\centerline{$(r_{{\minN}{\maxN}}=(You,\alpha) \land  r_{{\maxN}{\minN}}=(You,\beta)) \land  (\alpha,\beta)\in\{(0,0),(0,1),(1,1),(2,1),(2,2)\}$}

\end{definition}

\begin{definition}[toUpdate correct state]\label{def:toUpdateCS} 
Let $({\minN},{\maxN})$  be an edge with ${\minN} < {\maxN}$ in the state $(You,\alpha,\beta)$ in a configuration $C$. This state is said to be a \emph{toUpdate correct state}  if 
\begin{eqnarray*}
		\begin{split}
		[(\alpha,\beta)\in\{(0,1),(2,2)\} \land (r_{{\minN}{\maxN}}=(You,\alpha) \land  r_{{\maxN}{\minN}}=(You,\beta-1))]~ \\
		\lor ~ [(\alpha,\beta)\in\{(1,1),(2,1)\} \land (r_{{\minN}{\maxN}}=(You,\alpha-1) \land  r_{{\maxN}{\minN}}=(You,\beta))]~\\
		\end{split}
\end{eqnarray*}
\end{definition}

\begin{definition}[Correct state]\label{def:toUpdateCS} 
Let $({\minN},{\maxN})$  be an edge with ${\minN} < {\maxN}$ in the state $(You,\alpha,\beta)$ in a configuration $C$. This state is said to be \emph{correct} if the state is an updated or a toUpdate correct state.
\end{definition}

All four previous definitions deal with an edge in which the first node has a smaller identifier than the second node. In the following, we will write $(s,t)$ to denote such an edge. This constraint is due to the fact that nodes execute their $Increase$ rule one after the other in a specific order, and a link in state $(You, 0,1)$ can be correct while a link in state $(You,1,0)$ never is.
When we do not make any assumption on which node has the smallest identifier in an edge, we use the notation $(u,v)$. 

We now state that a node in an edge in a correct state is only activable for $Increase$ and $Write$ and that an edge in a correct state will forever remain in a correct state.
 


\begin{lemma}\label{lem:correctStateLP1}
Let $({\minN},{\maxN})$  be an edge with ${\minN} < {\maxN} $. Let $C$ be a configuration. If $(s,t)$ is in a correct state $(You, \alpha, \beta)$ in $C$ then 
neither $s$ nor $t$ is eligible for $Seduction(-)$, $Marriage(-)$ or $Reset$ in $C$;
\end{lemma}

\begin{proof}
Neither $s$ nor $t$ are eligible for a $Seduction(-)$ or a $Marriage(-)$ rule in $C$, since $p_s\neq null$ and $p_t\neq null$ in $C$. 
Since $p_s=t$, then $s$ is not eligible for a $x$-$Reset$ if $x\neq t$.  We now study the case of a $t$-$Reset$. $\PRaba(s)$ is False since $r_{ts}.p=You$. 
In $C$, we have: $(m_s, r_{ts}.m)\in \{ (0,0), (0,1), (1,1), (2,1), (2,2) \}$ and $r_{ts}.p=You$ and $s<t$. So, $\PRReset(s)$, does not hold in $C$. Thus, $s$ is not eligible for $Reset$ in $C$.
Since $p_t=s$, then $t$ is not eligible for a $x$-$Reset$ if $x\neq s$.  We now study the case of a $s$-$Reset$. $\PRaba(t)$ is False since $r_{st}.p=You$. 
In $C$, we have: $(m_s, r_{ts}.m)\in \{ (0,0), (1,0), (1,1), (1,2), (2,2) \}$ and $r_{st}.p=You$ and $s<t$. So, $\PRReset(t)$, does not hold in $C$. Thus, $t$ is not eligible for $Reset$ in $C$. \qed
\end{proof}

In fact, our algorithm is designed in such a manner  that, once an edge is in a correct state, it remains in a correct state forever.

\begin{lemma}\label{lem:correctStateLP2}
Let $({\minN},{\maxN})$  be an edge with ${\minN} < {\maxN} $. Let $C$ be a configuration. If $(s,t)$ is in a correct state $(You, \alpha, \beta)$ in $C$ then 
\begin{itemize}
\item $(s,t)$ forever remains in a correct state; 
\item  if neither $s$ nor $t$ is eligible for any rule in $C$, then $({\minN},{\maxN})$ is in the updated correct state  $(You, 2, 2)$ in $C$.
\end{itemize}
\end{lemma}

\begin{proof}
First, observe that the correct state definition only depends on the value of the following variables/registers $(p_s,p_t,m_s,m_t,r_{st},r_{ts})$.  Only $s$ and $t$ can write in these variables/registers. So, whatever a node $x\notin\{s,t\}$ executes in $C\mapsto C'$, this move cannot change the state of $(s,t)$ in $C'$. 
From Lemma~\ref{lem:correctStateLP1} $s$ and $t$ can only execute $Write$ or $Increase$,  
so the $p$-values will not change in $C \mapsto C'$ and then the state of $(s,t)$  in $C'$ only depends on the value of the following quadruplet $(m_s,m_t,r_{st}.m,r_{ts}.m)$ in $C'$.
Moreover, observe that if $s$ executes an $Increase$, then it is a $t$-$Increase$ since $p_s=t$ in $C$. In the same way, if $t$ executes an $Increase$, then it is a $s$-$Increase$ since $p_t=s$ in $C$. 
Finally, observe that if $s$ or $t$ executes a $Write(x)$ rule, with $x\notin \{s,t\}$, then this move cannot change the state of $(s,t)$ in $C'$. 

Thus, we are now going to perform a case study: we check for all possible correct state in $C$, which rules $s$ (resp. $t$) is eligible for among the rules $t$-$Increase$ and $Write(t)$ (resp. $s$-$Increase$ and $Write(s)$). e call these rules the \emph{relevant} rules since these are the only one that can change the state of $(s,t)$. For all possible transitions $C\mapsto C'$ that contains at least one of these rules, we prove that $(s,t)$ is in a correct state in $C'$. The two following tables present this case study. 
 
The $\alpha$ and $\beta$ values in $C$ are given in column 1. Column 2 gives the values of the quadruplet $(m_s,m_t,r_{st}.m, r_{ts}.m)$, according to the values of $\alpha$ and $\beta$. Columns 3 and 4 give rules $s$ and $t$ are eligible for. 

Observe that at each line, there is at most one node among $s$ and $t$ that is eligible for a relevant rule in $C$. If this node does not perform any rule in the transition starting in $C$, then the state of edge $(s,t)$ remains constant and  the proof is done. Otherwise, we obtain the considered configuration in the tables, called $D$. Thus, in column 5 and 6, we give the state that is reached after $s$ or $t$ performs its rule.   Observe that the last line of table 1 does not contain any value in the last two columns because there is no such a $D$ configuration since neither $s$ nor $t$ is eligible for a relevant rule. 

In the following table, we assume $(s,t)$ is in an updated correct state in $C$: (\emph{toUpdateCS} means toUpdate correct state)\\

\centerline{\small
\begin{tabular}[h]{|c||c|c|c||c|c|}\hline
\multicolumn{4}{|c||}{in $C$} & \multicolumn{2}{c|}{in $D$} \\\hline
\multirow{2}{*}{$(\alpha, \beta)$}{} & \multirow{2}{*}{$(m_s,m_t,r_{st},r_{ts})$}{} & \multicolumn{2}{c||}{relevant rules eligibility} & \multirow{2}{*}{$(m_s,m_t,r_{st},r_{ts})$}{} & \multirow{2}{*}{state of $(s,t)$}{}\\\cline{3-4}
						    &		& for $s$			& for $t$			&&\\\hline\hline
$(0,0)$ & $(0,0,0,0)$ & $\emptyset$ & $s$-Increase & $(0,1,0,0)$ & toUpdateCS $(You, 0,1)$\\\hline
$(0,1)$ & $(0,1,0,1)$ & $t$-Increase & $\emptyset$ & $(1,1,0,1)$ & toUpdateCS $(You, 1,1)$\\\hline
$(1,1)$ & $(1,1,1,1)$ & $t$-Increase & $\emptyset$ & $(2,1,1,1)$ & toUpdateCS $(You, 2,1)$\\\hline
$(2,1)$ & $(2,1,2,1)$ & $\emptyset$ & $s$-Increase & $(2,2,2,1)$ & toUpdateCS $(You, 2,2)$\\\hline
$(2,2)$ & $(2,2,2,2)$ & $\emptyset$ & $\emptyset$ & - & - \\\hline
\end{tabular}
}
~\\

In table 2, we assume $(s,t)$ is in a toUpdate correct state in $C$:  (\emph{updatedCS} means updated correct state)\\

\centerline{\small
\begin{tabular}[h]{|c||c|c|c||c|c|}\hline
\multicolumn{4}{|c||}{in $C$} & \multicolumn{2}{c|}{in $D$} \\\hline
\multirow{2}{*}{$(\alpha, \beta)$}{} & \multirow{2}{*}{$(m_s,m_t,r_{st},r_{ts})$}{} & \multicolumn{2}{c||}{relevant rules eligibility} & \multirow{2}{*}{$(m_s,m_t,r_{st},r_{ts})$}{} & \multirow{2}{*}{state of $(s,t)$}{}\\\cline{3-4}
						    &		& for $s$			& for $t$			&&\\\hline\hline
$(0,1)$ & $(0,1,0,0)$ & $\emptyset$ & $Write(s)$ & $(0,1,0,1)$ & updatedCS $(You, 0,1)$\\\hline
$(1,1)$ & $(1,1,0,1)$ & $Write(t)$ & $\emptyset$ & $(1,1,1,1)$ & updatedCS $(You, 1,1)$\\\hline
$(2,1)$ & $(2,1,1,1)$ & $Write(t)$ & $\emptyset$ & $(2,1,2,1)$ & updatedCS $(You, 2,1)$\\\hline
$(2,2)$ & $(2,2,2,1)$ & $\emptyset$ & $Write(s)$ & $(2,2,2,2)$ & updatedCS $(You, 2,2)$\\\hline
\end{tabular}
}
\end{proof}

\begin{corollary}\label{coro:correctStateLP}
Let $({\minN},{\maxN})$  be an edge with ${\minN} < {\maxN} $. Let $C$ be a configuration. If $(s,t)$ is in the correct state $(You, \alpha, \beta)$ in $C$ then 
neither $s$ nor $t$ is eligible for $Seduction(-)$, $Marriage(-)$ or $Reset$ from $C$.
\end{corollary}


%
%
%
%
%
%

\subsection{Correctness Proof }

\label{sec:correctness}

\begin{definition} \label{defMM1}
A configuration is called \emph{stable} if no node can execute a rule in this configuration.
\end{definition}

In particular, a configuration is stable iff all guards are false.

We now show that if our algorithm reaches a stable configuration then $p$-values define a maximal matching. The matching built  by the algorithm is $\mathcal M=\{ (u,v) \in E : p_u=v \land p_v=u\}$.


\begin{lemma}\label{lem:correct:1}
Let $u$ be a node. In any stable configuration \\
\centerline{$\forall v\in N(u) : r_{uv} =\FCTpointingto(u,v)$}
\end{lemma}

\begin{proof}
Let $C$ be a stable configuration. If $\exists v\in N(u) : r_{uv} \neq \FCTpointingto(u,v)$ in $C$, then $u$ is eligible for a $Write(v)$ rule and $C$ is not stable. 
\end{proof}

\begin{lemma}
\label{lem:correct:2}
Let $u$ be a node. In any stable configuration:\\
\centerline{$p_{u}=null \Rightarrow \forall v\in N(u) : p_{v} \not\in \{null,u\}$}
\end{lemma}

\begin{proof}
Let $C$ be a stable configuration where $p_u=null$. Consider a neighbor $v$. After Lemma~\ref{lem:correct:1}, $r_{uv} = (Idle,0)$.
 
If $p_{v}= null$ then we can assume without loss of generality that $u<v$. Then, according to Lemma~\ref{lem:correct:1}, $r_{vu} = (Idle,0)$. Thus $u$ is eligible for a $Seduction(v)$ rule and $C$ is not stable. 
 
If $p_{v}= u$ and $u<v$ then $\PRaba(v)$ holds since $r_{uv} \neq You$ and then $v$ is eligible for a $Reset$ rule and $C$ is not stable.
Finally, if $p_{v}= u$ and $v<u$ then, according to Lemma \ref{lem:correct:1}, $r_{vu}=(You,m_v)$. Then either $m_v=0$ and so $u$ is eligible for a $Marriage(v)$ rule, or $m_v\neq 0$ and so $\PRaba(v)$ holds and then $v$ is eligible for a $Reset$ rule. In both cases, $C$ is not stable.
\end{proof}
%

\begin{lemma}\label{obs:JC:edge}
Let $(s,t)$ be an edge with $s<t$. In any stable configuration, if $p_{s}=t$ and $p_{t}=s$ then edge $(s,t)$ is in the updated correct state $(You,2,2)$.
\end{lemma}

\begin{proof}
Consider the state $(You,\alpha,\beta)$ of edge $(s,t)$  in a stable configuration $C$. Observe that if an edge $(s,t)$ is in a correct state, then edge $(s,t)$ is in the updated correct state $(You,2,2)$  from Lemma~\ref{lem:correctStateLP2}, point 2.

We now prove by contradiction that this is the only possible case. Assume that the edge $(s,t)$ is not in a correct state. First observe that from Lemma~\ref{lem:correct:1}, we have $r_{st} =\FCTpointingto(s,t)$, 
and node ${s}$ (resp.  ${t}$) is not eligible for a $Write({t})$ (resp. $Write({s})$) rule.  This implies that if $(s,t)$ is in the state $(You,\alpha,\beta)$, then $r_{st}=(You,\alpha)$ and $r_{ts}=(You,\beta)$.
Thus, according to Definition~\ref{def:updatedCS}, the only remaining possibilities for $(\alpha,\beta)$ are $(2,0)$, $(0,2)$, $(1,0)$ and $(1,2)$. In all of these cases, $s$ is activable for a $Reset$ rule which contradicts the fact that $C$ is a stable configuration.

%
%
\end{proof}

\begin{lemma}
\label{lem:correct:3}
Let $(u,v)$ be an edge. In any stable configuration,  $p_{u}=v$ if and only if  $p_{v}=u$.
\end{lemma}

\begin{proof}

Assume, by contradiction, that $p_u=v$ and $p_v\neq u$. By Lemma~\ref{lem:correct:2}, $p_v\neq null$, thus $\exists v_1\in N(v) : p_v=v_1$ with $v_1 \neq u$. 
 
If $u>v$, after Lemma~\ref{lem:correct:1}, $r_{vu}.p = Other$ (\emph{i.e.}, $\neq You$). So the predicate $\PRaba(u)$ holds and node $u$ is eligible for a $Reset$ rule. Thus the configuration is not stable, which is impossible.
If $u<v$ and $m_v=2$ then Lemma~\ref{lem:correct:1} implies that $r_{vu} = (Other,2)$. Then the predicate $\PRaba(u)$ holds and $u$ is eligible for a $Reset$ rule. Thus the configuration is not stable neither. 

Thus $(p_u=v\land p_v\neq u) \Rightarrow (u<v\land m_v \in \{0,1\} \land \exists v_1\in N(v)\setminus\{u\} : p_v=v_1)$.

Suppose by contradiction $p_{v_1}=v$. From Lemma~\ref{obs:JC:edge}, the edge $(v_{1},v)$ is in updated correct state $(You,2,2)$. This implies that $m_{v}=2$ which contradicts the fact $m_v \in \{0,1\}$. So, $p_{v_1}\neq v$. Using the same argument for edge $(u,v)$, we can deduce: $v < v_1 \land m_{v_1}\in \{0,1\} \land \exists v_2\in N(v_1)\setminus\{v\} : p_{v_1}=v_2)$.
Now we can continue the construction in the same way.  We construct a path $(u, v, v_{1},  v_{2}, \ldots,v_r, \ldots)$  where $\forall i \geq 1: p_{v_i}=v_{i+1} \land v_{i} < v_{i+1} \land m_{v_i} \in \{0,1\}$. Since the number of nodes is finite, there exists a node $v_{y_1}$ that appears at least twice in the path. Thus, this path contains the cycle $(v_{y_1},v_{y_2}.\ldots,v_r,v_{y_1})$ and by construction, we have $v_r<v_{y_1}$  and $v_{y_1} < v_r$. This gives the contradiction.
\end{proof}

From these Lemmas, we deduce:

\begin{theorem}\label{th:correctness}
In any stable configuration, the set of edges $\mathcal{M}= \{ (u,v) \in E : p_u=v \land p_v=u\}$ is a maximal matching.
\end{theorem}

\begin{proof}
Let $C$ be a stable configuration.
By definition, the constructed set of edge $\mathcal{M}$ is a matching. 
From Lemma  \ref{lem:correct:3}, any  node $u$ such that $ p_u=v$ is in $\mathcal{M}$.
Lemma \ref{lem:correct:2}  implies that in $C$ no edge $(u,v)$  is such that  $ p_u=\bot$ and $ p_u=\bot$. So, $\mathcal{M}$ is maximal.

\end{proof}

\subsection{Convergence Proof}
\label{sec:convergence}

 The three following lemmas put in relation the number of moves of all rules except the $Write$ rule.

\subsection{Relationship between the all rules except the $Write$ rule}
  

\begin{lemma}\label{lem:JC:entre2mar}
Let $(\minN,\maxN)$ be an edge with $\minN<\maxN$. Let $\Exec$ be an execution containing two transitions $C_0 \mapsto C_1\mapsto^*D_0\mapsto D_1$ during which $\maxN$ executes a $Marriage(\minN)$ rule. 
Then  $\minN$ executes a $Seduction(\maxN)$ rule between $C_1$ and $D_0$.
\end{lemma}

\begin{proof}
From the $Marriage(\minN)$ rule, we have $r_{\minN\maxN} =(You,0)$ in $C_0$ and $D_0$. Moreover, according to Lemma~\ref{lem:u:increase}, $\maxN$ executes a $Reset$ rule between $C_1$ and $D_0$. 
Let $C_4\mapsto C_5$ be the transition where it does for the first time. Then $p_\maxN=\minN$ from $C_1$ to $C_4$.  

In the first step of the proof, we will prove that there exists two configurations $\gamma$ and $\gamma'$  between $C_0$ and $D_0$ and with $\gamma \mapsto^+ \gamma'$, such that  $(p_{\minN},m_{\minN})\neq(t,0)$ in $\gamma$ and $(p_{\minN},m_{\minN})=(t,0)$ in $\gamma'$. Then, we will prove that $s$ must execute a $Seduction(t)$ rule between $\gamma$ and $\gamma'$.
For the first step of the proof, we study two cases according to the value  of $r_{\minN\maxN}$  in $C_4$.

First, assume that $r_{\minN\maxN} \neq (You,0)$ in $C_{4}$.  So, there exists two transitions $C_2 \mapsto C_{3}$ and $C_6\mapsto C_7$, with $C_0 \mapsto^* C_2 \mapsto^+ C_4$ and 
$C_4 \mapsto^* C_6 \mapsto^+ D_0$ such that: (i) in $C_2 \mapsto C_3$, $\minN$ executes $Write(\maxN)$ to set $r_{\minN\maxN}$ to a couple $\neq (You,0)$ and so $(p_\minN,m_\minN)\neq (\maxN,0)$ in $C_2$; and (ii) in $C_6 \mapsto C_7$, $\minN$ executes $Write(\maxN)$ to set $r_{\minN\maxN}$ to $(You,0)$ and so $(p_\minN,m_\minN)=(\maxN,0)$ in $C_6$. We then have $\gamma=C_2$ and $\gamma'=C_6$. 


Second, if $r_{\minN\maxN} = (You,0)$ in $C_4$, then $\PRaba(t)$ is false and according to the $Reset$ rule and in particular to the $\PRReset(t)$ predicate, $m_{\maxN}= 2$ in $C_4$ (it is the only possible case as $s<t$ and $r_{st}.m=0$). Since $m_t\neq 2$ in $C_1$, there exists a transition $C_2\mapsto C_3$ between $C_1$ and $C_4$ where $\maxN$ executes an $Increase$ rule in order to write 2 in its $m$-variable. Since $p_\maxN=\minN$ from $C_1$ to $C_4$, then $p_t=s$ in $C_2$ and then, according to the $Increase$ rule, in $C_2$ $r_{\minN\maxN}= (You, 2)$.
Since $r_{\minN\maxN} =(You,0)$ in $C_0$ and $C_4$, then there exists two transitions $B_0 \mapsto B_1$ and $B_2\mapsto B_3$, with $C_0 \mapsto^* B_0 \mapsto^+ C_2$ and $C_2 \mapsto^* B_1 \mapsto^+ C_4$ such that: (i) in $B_0 \mapsto B_1$, $\minN$ executes $Write(\maxN)$ to set $r_{\minN\maxN}$ to $(You, 2)$ and so $(p_\minN,m_\minN)= (\maxN,2)$ in $B_0$; and (ii) in $B_2\mapsto B_3$, $\minN$ executes $Write(\maxN)$ to set $r_{\minN\maxN}$ to $(You,0)$ and so $(p_\minN,m_\minN)=(t,0)$ in $B_2$.  We then have $\gamma=B_0$ and $\gamma'=B_2$.

We now prove the second step of the proof, that is $s$ must execute a $Seduction(t)$ rule between $\gamma$ and $\gamma'$.
If $p_s\neq t$ in $\gamma$, and since $s<t$, then $s$ must execute a $Seduction(t)$ rule between $\gamma$ and $\gamma'$ in order to set $t$ in its $p$-variable. Otherwise, we have $p_s=t \land m_s\neq 0$ in $\gamma$. Let us assume that $s$ does not perform any $Seduction(t)$ rule between $\gamma$ and $\gamma'$. Thus, the only two rules to write 0 in its $m$-variable are $Marriage(-)$ and $Reset$. Since $s<t$, $s$ cannot execute a $Marriage(t)$ rule, thus after writing 0 in its $m$ variable, $p_s\neq t$ and we go back to the case 1, leading to the conclusion that $s$ must perform a $Seduction(t)$ rule in order to write $t$ in its $p$-variable.  Finally, in any cases, $s$ must execute a $Seduction(t)$ rule between $\gamma$ and $\gamma'$.
\end{proof}

\begin{lemma}\label{lem:u:BetweenTwoResets}
Let $u$ be a node. Let $\Exec$ be an execution where $u$ executes at least two $Reset$ moves. Let  $C_0 \mapsto C_1\mapsto^*C_2 \mapsto C_3$ be two transitions corresponding to two consecutive $Reset$ rule  executed by $u$.  Then ${u}$ executes a rule in $\{Seduction(-),Mariage(-)\}$ once  between $C_1$ and $C_2$.
\end{lemma}

\begin{proof}
 According to the $Reset$ rule, $p_u=null$ in $C_1$ and $p_u\neq null$ in $C_2$. so, $u$ has to execute a rule between $C_1$ and $C_2$ to set a neighbor identifier in its $p$-variable. There are only two rules doing that: the $Seduction$ and the $Marriage$ rules. Thus, $u$ executes such a rule at least once between $C_1$ and $C_2$. 
Now, assume that node ${u}$ executes such a rule more than once between $C_1$ and   $C_2$. Then, from Lemma \ref{lem:u:increase}, ${u}$ executes a $Reset$ rule between $C_{1}$ and $C_{2}$. This contradicts the fact that node $u$ does not execute any $Reset$ rule between $C_{1}$ and $C_{2}$.
Thus, the lemma holds. 
\end{proof}

\begin{lemma}\label{lem:u:3increase}
Let $u$ be a node. Let $\Exec$ be an execution where $u$ executes at least three $Increase$ moves. Let  $C_0 \mapsto C_1$, and $C_2 \mapsto C_3$ and $C_4 \mapsto C_5$  be three transitions corresponding to three consecutive $Increase$ rules  executed by $u$.  Then ${u}$ executes a  $Reset$ rule   once  between $C_0$ and $C_5$.
\end{lemma}

\begin{proof}
We now prove this lemma, by contradiction. Let us assume node ${u}$ does not executes an $Reset$ rule a  between $C_0$ and $C_5$. 

According to the $Increase$ rule, $m_u\in \{1,2\}$ in $C_3$ and $C_5$. According to Lemma~\ref{lem:u:increase}, if $m_u=1$ in $C_3$ (resp. $C_5$) then $u$ executes a $Reset$ rule between $C_1$ and $C_2$ (resp. $C_3$ and $C_4$). This contradicts the fact that node $u$ does not execute any $Reset$ rule between $C_{1}$ and $C_{2}$. Thus, $m_u= 2$ in $C_3$ and $C_5$. And so $m_u= 2$ in $C_3$ and $m_u=1$ in $C_4$. 

There is only one way to decrease the value of an $m$ variable: to write 0. Thus $u$ has to execute an $Reset$ rule between $C_3$ and $C_4$ in order to decrease the value  $m_u$. However this contradicts the fact that node $u$ does not execute any $Reset$ rule between $C_0$ and $C_5$.
\end{proof}

These lemmas bound the number of $Marriage$ (Lemma~\ref{lem:JC:entre2mar}), $Reset$  (Lemma~\ref{lem:u:BetweenTwoResets}) and $Increase$  (Lemma~\ref{lem:u:3increase}) in function of the number of $Seduction$. Then an upper bound on the number of $Write$ follows since one modification of the local variables of $u$ leads $u$ to execute at most $deg(u)$ $Write$.  So, in  the following, we present the sketch of the proof leading to an upper bound on the number of $Seduction$.
In the Lemma below, we state that the number of $Seduction(t)$ by $s$ is strongly connected to the number of times that $t$ writes $2$ in $m_{t}$. 

\begin{lemma}\label{lemma:nb:seduction}
Let $({\minN},{\maxN})$ be an edge with ${\minN}<{\maxN}$. 
Let $\Exec$ be an execution containing three transitions $D_0 \mapsto D_1\mapsto^*D_{2}\mapsto D_{3}\mapsto^*D_{4}\mapsto D_{5}$ where $\minN$ executes a $Seduction(\maxN)$ rule. Then there exists a transition $D\mapsto D'$ between $D_2$ and $D_4$ where $\maxN$ executes a $Write(\minN)$ rule and with  in $D$: $p_t\neq null$ and $m_t=2$.
\end{lemma}

From the previous Lemma, we know that $t$ has to write $2$ in its $m$-variable for $s$ to reset a previous $Seduction(t)$. The next Theorem gives conditions where the value $2$ in a $m$-variable corresponds to a locked marriage and thus yields  to a situation where $s$ cannot seduce $t$ anymore.

\begin{theorem}\label{lem:mvaut2}
Let $(u,v)$ be an edge. Let $\Exec$ be an execution. 
If $\Exec$ contains two transitions $A_0\mapsto A_1\mapsto^*A_2\mapsto A_3\mapsto^*A_4$ such that : 
\begin{itemize} 
\item in $A_0\mapsto A_1$, $v$ executes a $u$-rule;
\item in $A_2\mapsto A_3$, $u$ executes a $Reset$ rule;
\item and in $A_4$, $(p_u,m_u)=(v,2)$;
\end{itemize}
then the edge $(min(u,v),max(u,v))$ is in a correct state in $A_4$.
\end{theorem}

\subsection{Number of moves for the  $Seduction$ Rule}

\begin{theorem}\label{theo:nb:seduction}
Let $({\minN},{\maxN})$ be an edge with ${\minN}<{\maxN}$.  The number of $Seduction({\maxN})$ rules executed by node ${{\minN}}$ is in  $O(\Delta)$.
\end{theorem}

\begin{proof}
By contradiction. Let $\Exec$ be an execution where $s$ executes a $Seduction(t)$ rule at least $2\Delta+4$ times. We are going to show that such an execution is not possible since after the $(2\Delta+3)^{th}$ $Seduction(t)$ execution, $s$ cannot perform any other $Seduction(t)$ rule. 

For $0\leq i\leq (2\Delta+3)$, let $A_i\mapsto B_i$  be the transition where $s$ executes its $i^{th}$ $Seduction(t)$ rule in $\Exec$.

According to Lemma~\ref{lemma:nb:seduction}, between each couple of configurations $(B_j, A_{j+1})$ where $1\leq j\leq (2\Delta+1)$, there exists a transition $C_j\mapsto D_j$ where $t$ executes a $Write(s)$ rule and with $p_t\neq null$ and $m_t=2$ in $C_j$. 

Since $t$ has at most $\Delta$ neighbors ans since there are $2\Delta+1$ such transitions $C_j\mapsto D_j$, then there exists at least one neighbor of $t$ that appears at least 3 times in $p_t$ among these $C_j$. 
More formally, $\exists x\in \Voisin(t) : \exists \text{ distinct } a,b,c \in [1,..,2\Delta+1] :: p_t=x$ in $C_a$, $C_b$ and $C_c$. Let us assume w.l.o.g that $a<b<c$.
 
 First, let us prove that $x\neq s$. By contradiction. The we consider the three transitions $A_0\mapsto B_0$, $A_a\mapsto B_a$ and $A_{a+1}\mapsto B_{a+1}$. Observe that the execution starting in configuration $A_a$ reaches all the assumptions made in Lemmas~\ref{grand2}. Indeed, before $A_a$, both $s$ executes a $t$rule ($Seduction(t)$) and $t$ executes a Reset by Lemma~\ref{lem:LP:entre2sed}. Thus, when $(p_t,m_t)=(s,2)$ in configuration $C_a$, the edge $(s,t)$ is in a correct state. Thus, by Corollary~\ref{coro:correctStateLP}, from this configuration $s$ cannot execute any $Seduction$, which contradict the fact that it does in transition $A_{b+1}\mapsto B_{b+1}$. So $x\neq s$. 
 
According to Lemma~\ref{lem:jc:reset}, between $C_a$ and $C_b$, $x$ executes a $Write(t)$ rule and between $C_b$ and $C_c$, $t$ executes a $Reset$ rule. Finally, in $C_c$, we have $(p_t,m_t)=(x,2)$. Thus by Theorem~\ref{lem:mvaut2}, the edge $(min(t,x), max(t,x))$ is in a correct state in $C_c$. From Corollary~\ref{coro:correctStateLP}, from this configuration $t$ cannot execute any $Reset$. However, since $s$ executes two $Seduction(t)$ rules in $A_{2\Delta+2}\mapsto B_{2\Delta+2}$  and $A_{2\Delta+3}\mapsto B_{2\Delta+3}$ and by Lemma~\ref{lem:LP:entre2sed}, $t$ executes a $Reset$ between $B_{2\Delta+2}$ and $A_{2\Delta+3}$ which leads the contradiction
\end{proof}

\begin{lemma}\label{lem:u:increase}
Let $u$ be a node. Let $\Exec$ be an execution containing two transitions $C_0 \mapsto C_1$ and $C_2 \mapsto C_3$ with $C_{1} \mapsto^* C_{2}$ where ${u}$ executes a rule. 
\begin{enumerate} 
\itemshort
\item If ${u}$ executes an $Increase$ rule during these two transitions and if $m_{u} = 1$ in $C_{3}$, then ${u}$ executes a $Reset$ rule between $C_{1}$ and $C_{2}$.
\item If ${u}$ executes a $Seduction(-)$ or a $Mariage(-)$ rule during these two transitions, then ${u}$ executes a $Reset$ rule between $C_{1}$ and $C_{2}$.
\end{enumerate}
\end{lemma}

\begin{proof}
We start by proving the first point. According to the $Increase$ rule, $p_u\neq null$ in $C_0$ and $u$ writes $1$ or $2$ in $m_{u}$ during the transition $C_0 \mapsto C_1$. So $m_u\neq 0$ and $p_u\neq null$ in $C_1$. Since $m_{u} = 1$ in $C_{3}$, then $m_{u} = 0$ in $C_{2}$. Thus either $u$ sets its $m$ variable to 0 executing a $Reset$ rule between $C_1$ and $C_2$ and the proof is done, or $u$ executes a $Marriage(-)$ or a $Seduction(-)$ rule, let say in transition $C\mapsto C'$. Then we have $p_u=null$ in $C$. Since $p_u\neq null$ in $C_1$, then $u$ must execute a $Reset$ rule between $C_1$ and $C$. 

We now prove the second point. According to both rules $Seduction$ and $Marriage$, $p_{u}\neq \bot$ in $C_1$ and $p_{u} = \bot$ in $C_2$. Thus $u$ must execute a $Reset$ rule between $C_1$ and $C_2$ in order to set $p_{u}$ to~$\bot$.
\end{proof}


\begin{lemma}\label{lem:LP:entre2sed}
Let $(\minN,\maxN)$ be an edge with $\minN<\maxN$. Let $\Exec$ be an execution containing two transitions $C_0 \mapsto C_1$ and $D_0\mapsto D_1$ with $C_{1} \mapsto^* D_0$ where $\minN$ executes a $Seduction(\maxN)$ rule. 
We have: both $\minN$ and $\maxN$ execute $Reset$ between $C_1$ and $D_0$.
\end{lemma}

\begin{proof}
From the $Seduction(\maxN)$ rule, we have $r_{\maxN\minN} =(Idle,0)$ in $C_0$ and $D_0$. Moreover, according to Lemma~\ref{lem:u:increase}, $\minN$ executes a $Reset$ rule between $C_1$ and $D_0$. 
Let $C_4\mapsto C_5$ be the transition where it does for the first time. 

We now study two cases: in $C_4, r_{\maxN\minN}$ is either equal or different from $(Idle,0)$.

If it is different then there exists two transitions $C_2 \mapsto C_{3}$ and $C_6\mapsto C_7$, with $C_0 \mapsto^* C_2 \mapsto^+ C_4$ and 
$C_4 \mapsto^* C_6 \mapsto^+ D_0$ such that: (i) in $C_2 \mapsto C_3$, $\maxN$ executes $Write(\minN)$ to set $r_{\maxN\minN}$ to a couple $\neq (Idle,0)$ and so $p_\maxN \neq null$ in $C_2$; and (ii) in $C_6 \mapsto C_7$, $\maxN$ executes $Write(\minN)$ to set $r_{\maxN\minN}$ to $(Idle,0)$ and so $p_\maxN=null$ in $C_6$.
Thus, there exists a transition between $C_3$ and $C_6$ where $\maxN$ executes a $Reset$ move. 

Now, if $r_{\maxN\minN}=(Idle,0)$ in $C_4$. Then, $\PRReset(s)$ is false and according to the $Reset$ rule and in particular to the $\PRaba(s)$ predicate, $m_{\minN}\neq 0$ in $C_4$. Moreover, from the $Seduction(\maxN)$ rule, $m_s=0$ in $C_1$. 
Thus there exists a transition $C_2\mapsto C_3$ between $C_1$ and $C_4$ where $\minN$ executes an $Increase$ rule. Since $\minN$ executes its first $Reset$ from $C_1$ in $C_4\mapsto C_5$, then $p_\minN=\maxN$ from $C_1$ to $C_4$, and so $p_s=t$ in $C_2$. According to the $Increase$ rule, in $C_2$ $r_{\maxN\minN}= (You, 1)$.
Since $r_{\maxN\minN} =(Idle,0)$ in $C_0$ and $C_4$, then there exists two transitions $B_0 \mapsto B_1$ and $B_2\mapsto B_3$, with $C_0 \mapsto^* B_0 \mapsto^+ C_2$ and $C_2 \mapsto^* B_2 \mapsto^+ C_4$ such that: (i) in $B_0 \mapsto B_1$, $\maxN$ executes $Write(\minN)$ to set $r_{\maxN\minN}$ to $(You, 1)$ and so $(p_\maxN,m_\maxN)\neq (\minN,1)$ in $B_0$; and (ii) in $B_2\mapsto B_3$, $\maxN$ executes $Write(\minN)$ to set $r_{\maxN\minN}$ to $(Idle,0)$ and so $(p_\maxN,m_\maxN)=(\bot,0)$ in $B_2$. Thus, there exists a transition between $B_1$ and $B_2$ where $\maxN$ executes a $Reset$ move. 
\end{proof}

\begin{lemma}\label{lem:technical}
 Let $({\minN},{\maxN})$ be an edge with ${\minN}<{\maxN}$. Let $\Exec$ be an execution. If $\Exec$ contains two configurations $L_{0}$ and $L_{1}$ with $L_{0} \mapsto^* L_{1}$ and such that:
\begin{itemize}
\item ${\minN}$ executed at least one ${\maxN}$-rule before $L_{0}$; 
\item $r_{{\minN}{\maxN}}=(You,0)$ in $L_{0}$;
\item  $(p_{\minN},m_{\minN})=({\maxN},2)$ in $L_1$;
\end{itemize}
then, ${\minN}$ executes a $t$-$Increase$ to set $m_{\minN}=2$ between $L_0$ and $L_1$.
\end{lemma}

\begin{proof} 
In $L_0$ there are two cases concerning the value of $p_s$: either $p_s=t$ or not. We consider the first case where $p_s=t$.
Let $D_0 \mapsto D_1$ be the last transition before $L_0$ in which $s$ executes a $t$-rule. Thus we have $r_{st}=(You,0)$  and $p_s=t$ in $D_1$.

According to the $t$-rule $s$ executes in $D_0 \mapsto D_1$, we can deduce its $(p_s,m_s)$ values in $D_1$.
\begin{itemize}
\item $Write(t)$: $(p_s,m_s)=(t,0)$ in $D_1$;
\item $t$-$Reset$: $(p_s,m_s)=(null,0)$ in $D_1$. This is not possible since $s$ cannot execute any $t$-rule between $D_1$ and $L_0$ and since $p_s=t$ in $L_0$;
\item $Marriage(t)$: not possible since $t>s$.
\end{itemize}

For the last two rules, since $r_{st}=(You, 0)$ in $D_1$ and since $s$ does not perform a $Write(t)$ in $D_0\mapsto D_1$, then $r_{st}=(You, 0)$ in $D_0$. Moreover, according to the $Seduction(t)$ and the $t$-$Increase$ rules, $r_{st}=\FCTpointingto(s,t)$ in $D_0$ if $s$ executes one of these two rules in $D_0\mapsto D_1$. Thus,  $(p_s, m_s)=(t,0)$ in $D_0$.

\begin{itemize} 
\item $Seduction(t)$: not possible since we should have $p_s=null$ in $D_0$;
\item $t$-$Increase$:  $(p_s,m_s)=(t,1)$ in $D_1$.
\end{itemize}

Thus $(p_s,m_s) \in \{(t,0),(t,1)\}$ in $D_1$. Observe now that $s$ cannot execute a $x$-$Increase$ rule for any node $x$, between $D_1$ and $L_0$: by construction it cannot execute it for node $t$. Also, it cannot execute it for any other node, since $p_s=t$ in $D_1$ and it cannot be modified between $D_1$ and $L_0$ (because $s$ cannot execute any $t$-rule between these two configurations). Thus $s$ is not eligible for an $Increase$ rule between $D_1$ and $L_0$. 
We obtain $m_s \in \{0,1\}$ in $L_0$ and $(p_s,m_s) =(t,2)$ in $L_1$. Thus $s$ must execute a $t$-$Increase$ to set $m_s=2$ between $L_0$ and $L_1$ and the proof is done.

We now study the second case where $p_s \neq t$ in $L_0$. Since in $L_1$, $p_s=t$ by assumption, then $s$ must execute a $Seduction(t)$ rule in some transition $C_0 \mapsto C_1$ between $L_0$ and $L_1$ and so $m_s=0$ in $C_1$. Since $(p_s,m_s) =(t,2)$ in $L_1$, then $s$ must execute a $t$-$Increase$ to set $m_s=2$ between $C_1$ and $L_1$. Finally, the proof is done because $L_0\mapsto^+ C_1 \mapsto^+L_1$.
\end{proof}


\begin{lemma}\label{lem1:technical}
Let $({\minN},{\maxN})$ be an edge with ${\minN}<{\maxN}$. Let $\Exec$ be an execution. If $\Exec$ contains two configurations $L_0$ and $L_1$ with $L_0\mapsto^* L_1$ and such that:
\begin{itemize} 
\item ${\maxN}$ executed at least one ${\minN}$-rule before $L_{0}$;
\item $r_{{\maxN}{\minN}}=(Idle,0)$ in $L_{0}$;
\item  $(p_{\maxN},m_{\maxN})=({\minN},1)$ in $L_1$;
\end{itemize}
then, ${\maxN}$ executes a $s$-$Increase$ to set $m_{\minN}=1$ between $L_0$ and $L_1$.
\end{lemma}

\begin{proof}
In $L_0$ there are two cases concerning the value of $p_t$: either $p_t=s$ or not. We consider the first case where $p_t=s$.
Let $D_0 \mapsto D_1$ be the last transition before $\Exec$ in which $t$ executes a $s$-rule. Thus we have $r_{ts}=(Idle,0)$  and $p_t=s$ in $D_1$.

According to the $s$-rule $t$ executes in $D_0 \mapsto D_1$, we can deduce its $m_t$ value in $D_1$:
\begin{itemize} 
\item $Write(u)$: not possible since this would imply $p_t=s$ in $D_0$ and then $r_{ts}.p=You$ in $D_1$. But $r_{ts}=(Idle,0)$ in $D_1$.
\item  $Reset(s)$: not possible since this would imply $p_t=\bot$ in $D_1$.
\item $Increase(s)$: not possible since this would implies in $D_0$: $\neg PRwriting(s)$ and $p_s=t$ and $r_{ts}=(Idle,0)$ (since only $m_s$ would be modified). And these three conditions are not compatible.  
\item $Seduction(s)$: $t$ cannot perform this rule since $s<t$.
\item $Marriage$: $(p_t,m_t)=(s,0)$ in $D_1$.
\end{itemize}

Thus $(p_t,m_t)=(s,0)$ and $r_{ts}=Idle,0$ in $D_1$.  Since this is the last $s$-rule $t$ executes before $L_0$ then observe that $p_t$ cannot be modified. Thus between $D_1$ and $L_0$ $t$ cannot execute $Seduction$, $Marriage$ or $Reset$ rules. Since $t$ is not eligible for $Write(s)$ then $r_{ts}$ remains equal to $(Idle,0)$ until $L_0$. Observe also that since $p_t=s$ between $D_1$ and $L_0$, $t$ cannot execute $Increase(x)$ for any node $x$. Thus $m_t=0$ remains $True$ between these configurations. This implies that $m_t=0$ in $L_0$.  We obtain that $m_t=0$ in $L_0$ and since $(p_t,m_t)=(s,1)$ in $L_1$ then $t$ must execute an $Increase(s)$ writing $m_t:=1$ between $L_0$ and $L_1$.

We now study the second case where $p_t \neq s$ in $L_0$. Since in $L_1$, $p_t=s$ by assumption, then $t$ must execute a $Marriage(s)$ rule in some transition $C_0 \mapsto C_1$ between $L_0$ and $L_1$ and so $m_t=0$ in $C_1$. Since $(p_t,m_t) =(s,1)$ in $L_1$, then $t$ must execute a $s$-$Increase$ to set $m_t:=1$ between $C_1$ and $L_1$. Finally, the proof is done because $L_0\mapsto^+ C_1 \mapsto^+L_1$.
\end{proof}

\begin{lemma}\label{grand2}
Let $({\minN},{\maxN})$ be an edge with ${\minN}<{\maxN}$. Let $\Exec$ be an execution. 
If $\Exec$ contains two transitions $A_0\mapsto A_1$, $A_2\mapsto A_3$ and a configuration $A_4$, with $A_1\mapsto^* A_2$ and $A_3 \mapsto^* A_4$  and such that : 
\begin{itemize} 
\item in $A_0\mapsto A_1$, $s$ executes a $t$-rule;
\item in $A_2\mapsto A_3$, $t$ executes a $Reset$ rule;
\item and in $A_4$, $(p_t,m_t)=(s,2)$;
\end{itemize}
then the edge $(s,t)$ is in a correct state in $A_4$.
\end{lemma}

\begin{proof}
Let $C_0 \mapsto C_1$ be the last $Reset$ executed by $t$ between $A_2$ and $A_4$ (we can have $C_0=A_2$). In $C_1$, $p_t=null$ and in $A_4$, $p_t=s$, with $t>s$. Thus $t$ must execute a $Marriage(s)$ rule between $C_1$ and $A_4$ to set $p_t=s$

Let $C_2 \mapsto C_3$ be the last $Marriage(-)$ rule executed by $t$ between $C_1$ and $A_4$. Since $t$ does not perform any $Reset$ from $C_3$ to $A_4$ by construction, then $p_t$ remains contant from $C_3$ to $A_4$. $p_t=s$ in $A_4$, thus $p_t=s$ in $C_3$ and so $t$ performs a $Marriage(s)$ rule in $C_2 \mapsto C_3$.  

Observe that between $C_3$ and $A_4$, we have by construction: $t$ does not perform any $Reset$ nor $Marriage$ and $p_t=s$. Thus, $t$ cannot perform any $Seduction$ rule neither. So, the only rule $t$ can perform between $C_3$ and $A_4$ are $Write(-)$ and $s$-$Increase$ (since $p_t=s$). So the value of $m_t$ can only change by a +1 incrementation between $C_3$ and $A_4$ . 
In $C_3$, $m_t=0$ and in $A_4$, $m_t=2$. Thus, beside the $Write$ rule, $t$ executed exactly two $s$-$Increase$ between $C_3$ and $A_4$. 

Let $C_4 \mapsto C_5$ and $C_6 \mapsto C_7$ be these two $s$-$Increase$ executed by $t$, with $C_5\mapsto^* C_6$. In $C_4 \mapsto C_5$, $t$ sets $m_t=1$ and in $C_6 \mapsto C_7$, $t$ sets $m_t=2$. So in $C_4$, $m_t=0$ and in $C_6$, $m_t=1$. According to the $Increase$ rule, we have: in $C_4$, $r_{st}=(You, 0)$ and in $C_6$, $r_{st}=(You, 2)$. Thus, $s$ performs a $Write(t)$ rule between $C_4$ and $C_6$ to set $r_{st}=(You, 2)$. 

Let $D_0 \mapsto D_1$ be this transition. We thus have $(p_s,m_s)=(t,2)$ in $D_0$ with $C_4\mapsto^* D_0$. 

From Lemma~\ref{lem:technical} -- by setting $C_4=L_0$ and $D_0=L_1$ and considering that $s$ executed a $t$-rule in $A_0\mapsto A_1$ -- there exists a transition $F_0\mapsto F_1$ between $C_4$ and $D_0$ where $s$ executes a $t$-$Increase$ rule to set $m_s=2$.

We are now going to prove that the edge $(s,t)$ is in the updated correct state $(You, 1,1)$ in $F_0$. 

Since $s$ executes a $t$-$Increase$ in $F_0\mapsto F_1$ with $m_s=2$ in $F_1$, then in $F_0$ we have: $(p_s,m_s)=(t,1)$ and, according to the $Increase$ rule, $r_{ts}=(You,1)$. We also have $r_{st}=(You,1)$ otherwise $s$ would have executed a $Write(t)$ instead of an $Increase$ in $F_0\mapsto F_1$.

In $C_4\mapsto C_5$, $t$ executes a $s$-Increase setting $m_t=1$, so $(p_t,m_t)=(s,0)$ in $C_4$ and $(p_t,m_t)=(s,1)$ in $C_5$. Moreover, $t$ can only execute some $Write(-)$ rules between $C_5$ and $C_6$. Thus in $C_4$, $(p_t,m_t)=(s,0)$ and from $C_5$ to $C_6$, $(p_t,m_t)=(s,1)$. 

We have: $C_4\mapsto^* F_0\mapsto^* D_0 \mapsto^+ C_6$ and we know that $r_{st}=(You, 0)$ in $C_4$ and $r_{st}=(You, 1)$ in $F_0$. So $C_4\neq F_0$ and so $C_5\mapsto^* F_0\mapsto^* D_0 \mapsto^+ C_6$. So $(p_t,m_t)=(s,1)$ in $F_0$. 

We finally obtain for $F_0$: $(p_s,m_s)=(t,1)$, $(p_t,m_t)=(s,1)$, $r_{st}=(You, 1)$ and $r_{ts}=(You,1)$. Thus the edge $(s,t)$ is in the updated correct state $(You, 1,1)$ in $F_0$. As $F_0\mapsto^+C_6\mapsto^+A_4$ and according to lemma~\ref{lem:correctStateLP2}, the edge $(s,t)$ is in a correct state in $A_4$.
\end{proof}


\begin{lemma}\label{petit1}
Let $({\minN},{\maxN})$ be an edge with ${\minN}<{\maxN}$. Let $\Exec$ be an execution. 
If $\Exec$ contains two transitions $A_0\mapsto A_1$, $A_2\mapsto A_3$ and a configuration $A_4$, with $A_1\mapsto^* A_2$ and $A_3 \mapsto^* A_4$  and such that : 
\begin{itemize} 
\item in $A_0\mapsto A_1$, $t$ executes a $s$-rule;
\item in $A_2\mapsto A_3$, $s$ executes a $Reset$ rule;
\item and in $A_4$, $(p_s,m_s)=(t,1)$;
\end{itemize}
then the edge $(s,t)$ is in a correct state in $A_4$.
\end{lemma}

\begin{proof}
Let $C_0 \mapsto C_1$ be the last $Reset$ executed by $s$ between $A_2$ and $A_4$ (we can have $C_0=A_2$). In $C_1$, $p_s=null$ and in $A_4$, $p_s=t$, with $s<t$. Thus $s$ must execute a $Seduction(t)$ rule between $C_1$ and $A_4$ to set $p_s=t$

Let $C_2 \mapsto C_3$ be the last $Seduction(-)$ rule executed by $s$ between $C_1$ and $A_4$. Since $s$ does not perform any $Reset$ from $C_3$ to $A_4$ by construction, then $p_s$ remains contant from $C_3$ to $A_4$. $p_s=t$ in $A_4$, thus $p_s=t$ in $C_3$ and so $s$ performs a $Seduction(t)$ rule in $C_2 \mapsto C_3$.  

Observe that between $C_3$ and $A_4$, we have by construction: $s$ does not perform any $Reset$ nor $Seduction$ and $p_s=t$. Thus, $s$ cannot perform any $Marriage$ rule neither. So, the only rule $s$ can perform between $C_3$ and $A_4$ are $Write(-)$ and $t$-$Increase$ (since $p_s=t$). So the value of $m_s$ can only change by a +1 incrementation between $C_3$ and $A_4$ . 
In $C_3$, $m_s=0$ and in $A_4$, $m_s=1$. Thus, beside the $Write$ rule, $s$ executed exactly one $t$-$Increase$ between $C_3$ and $A_4$. 

Let $C_4 \mapsto C_5$ be this $t$-$Increase$ executed by $s$. In $C_4 \mapsto C_5$, $s$ sets $m_s=1$. So, in $C_4$, $(p_s,m_s)=(t,0)$ and then $r_{ts}= (You, 1)$. Moreover, according to the $Seduction(t)$ rule, in $C_2$,  $r_{ts}= (Idle, 0)$. So there exists a transition $D_0\mapsto D_1$ between $C_2$ and $C_4$ where $t$ executes a $Write(s)$ rule to set $r_{ts}=(You, 1)$. Thus, in $D_0$, $(p_t,m_t)=(s,1)$.

From Lemma~\ref{lem1:technical} -- by setting $C_2=L_0$ and $D_0=L_1$ and considering that $t$ executed a $s$-rule in $A_0\mapsto A_1$ -- there exists a transition $F_0\mapsto F_1$ between $C_2$ and $D_0$ where $t$ executes a $s$-$Increase$ rule to set $m_t=1$.

We are now going to prove that the edge $(s,t)$ is in the updated correct state $(You, 0,0)$ in $F_0$. 

Since $t$ executes a $s$-$Increase$ in $F_0\mapsto F_1$ with $m_t=1$ in $F_1$, then in $F_0$ we have: $(p_t, m_t) = (s, 0)$ and, according to the $Increase$ rule, $r_{st}=(You,0)$. We also have $r_{ts}=(You, 0)$ otherwise $t$ would have perform a $Write(s)$ instead of an $s$-$Increase$ in $F_0\mapsto F_1$. 

Recall, that $C_2 \mapsto^* F_0 \mapsto^+ D_0$. We already know that $r_{ts}=(Idle, 0)$ in $C_2$, so $C_3 \mapsto^* F_0 \mapsto^+ D_0$.

Moreover, we know that between $C_3$ and $C_4$, $s$ does not perform any rule but some $Write(-)$ rule. So $(p_s,m_s)$ remains constant between $C_3$ and $C_4$. $s$ executes a $Seduction(t)$ rule in $C_2\mapsto C_3$, so $(p_s,m_s)=(t,0)$ in $C_3$. Moreover, $C_3\mapsto^* D_0 \mapsto^+ C_4$, so $(p_s,m_s)=(t,0)$ in $D_0$. 

We finally obtain for $F_0$: $(p_t,m_t)=(s,0)$, $(p_s,m_s)=(t,0)$, $r_{ts}=(You, 0)$ and $r_{st}=(You,0)$. Thus the edge $(s,t)$ is in the updated correct state $(You, 0,0)$ in $F_0$. As $F_0\mapsto^+C_4\mapsto^+A_4$ and according to lemma~\ref{lem:correctStateLP2}, the edge $(s,t)$ is in a correct state in $A_4$.

\end{proof}

\begin{theorem}{lem:mvaut2}
Let $(u,v)$ be an edge. Let $\Exec$ be an execution. 
If $\Exec$ contains two transitions $A_0\mapsto A_1$, $A_2\mapsto A_3$ and a configuration $A_4$, with $A_1\mapsto^* A_2$ and $A_3 \mapsto^* A_4$  and such that : 
\begin{itemize} 
\item in $A_0\mapsto A_1$, $v$ executes a $u$-rule;
\item in $A_2\mapsto A_3$, $u$ executes a $Reset$ rule;
\item and in $A_4$, $(p_u,m_u)=(v,2)$;
\end{itemize}
then the edge $(min(u,v),max(u,v))$ is in a correct state in $A_4$.
\end{theorem}

\begin{proof}
 If $u>v$, we conclude immediately using Lemma~\ref{grand2}. Thus assume that $u<v$.
 
  Let $C_0 \mapsto C_1$ be the last $Reset$ executed by $u$ between $A_2$ and $A_4$ (we can have $C_0=A_2$). In $C_1$, $p_u=null$ and in $A_4$, $p_u=v$, with $u<v$. Thus $u$ must execute a $Seduction(v)$ rule between $C_1$ and $A_4$ to set $p_u=v$

Let $C_2 \mapsto C_3$ be the last $Seduction(-)$ rule executed by $u$ between $C_1$ and $A_4$. Since $u$ does not perform any $Reset$ from $C_3$ to $A_4$ by construction, then $p_u$ remains constant from $C_3$ to $A_4$. Since $p_u=v$ in $A_4$, then $p_u=v$ in $C_3$ and so $u$ performs a $Seduction(v)$ rule in $C_2 \mapsto C_3$. 

Observe that between $C_3$ and $A_4$, we have by construction: $u$ does not perform any $Reset$ nor $Seduction$ and $p_u=v$. Thus, $u$ cannot perform any $Marriage$ rule neither. So, the only rule $u$ can perform between $C_3$ and $A_4$ are $Write(-)$ and $v$-$Increase$ (since $p_u=v$). So the value of $m_u$ can only change by a +1 incrementation between $C_3$ and $A_4$ . 
In $C_3$, $m_u=0$ and in $A_4$, $m_u=2$. Thus, beside the $Write$ rule, $u$ executed exactly two $v$-$Increase$ between $C_3$ and $A_4$. Let $C_4 \mapsto C_5$ be the transition in which $u$ executes the first such $v$-$Increase$ rule with $C_4 \geq C_3$. By definition of this rule, we have that in $C_5$, $(p_u,m_u)=(v,1)$ holds. We now apply Lemma~\ref{petit1}. Observe that $C_5$ is after transition $A_0 \mapsto A_1$ in which $v$ executes a $u$-rule and after $A_2 \mapsto A_3$ in which $u$ executes a $Reset$ rule. Thus applying Lemma~\ref{petit1} with transitions $A_1 \mapsto A_2$ and $A_3 \mapsto A_4$ and configuration $C_5$ we get that edge $(u,v)$ is in a correct state in $C_5$. By Lemma~\ref{lem:correctStateLP2}, this implies that configuration $A_4$ is in a correct state.

\end{proof}

\begin{lemma}\label{lem:lp:write}
Let $({\minN},{\maxN})$ be an edge with ${\minN}<{\maxN}$.
Let $\Exec$ be an execution containing the two following transitions:
\begin{itemize} 
\item $C_0\mapsto C_1$ where $s$ executes a $Seduction(t)$ rule;
\item $C_2\mapsto C_3$ where $s$ executes a $t$-$Reset$ rule; 
\item with $C_1\mapsto^* C_2$ and with $s$ does not execute any $Reset$ between $C_1$ and $C_2$. 
\end{itemize}
Then, $t$ executes a $Write(s)$ rule between $C_0$ and $C_2$.
\end{lemma}

\begin{proof}
Since $s$ executes a $Seduction(t)$ rule in $C_0\mapsto C_1$, then $r_{ts}=(Idle, 0)$ in $C_0$. 
Since $s$ executes a $t$-$Reset$ rule in $C_2\mapsto C_3$, then $r_{ts}=(Idle, 0)$ in $C_0$. 
Either $r_{ts}\neq (Idle, 0)$ in $C_2$ and so the proof is done, or $r_{ts} = (Idle, 0)$ in $C_2$. 

In the second case, $m_s=0$ in $C_1$ (by the $Seduction$ rule) and $m_s\neq 0$ in $C_2$. So, $s$ executes an $Increase$ rule between $C_1$ and $C_2$. Let $C_3\mapsto C_4$ be the transition where $s$ does so for the first time. In $C_1$, $p_s=t$ and $m_s=0$, by the $Seduction$ rule. Moreover, by assumption, $s$ does not execute any $Reset$ from $C_1$ to $C_2$. Thus, $s$ can only perform some $Write(-)$ rule or a $t$-$Increase$ rule from $C_1$ to $C_2$. And so, in $C_3$, $p_s=t$ and $m_s=0$ (since $C_3\mapsto C_4$ is the first $Increase$ from $s$). $s$ performs then $m_s:=1$ in $C_3\mapsto C_4$ and so $r_{ts}=(You, 1)$ in $C_3$. Since $r_{ts}=(Idle, 0)$ in $C_0$, the proof is done. 
\end{proof}

\begin{lemma}\label{lem:jc:reset}
Let  $(\minN,\maxN)$ be an edge. Let $\Exec$ be an execution containing two transitions $A_0 \mapsto A_1$ and $A_{2}\mapsto A_{3}$ with $A_{1} \mapsto^* A_{2}$ where $\maxN$ executes a $write(\minN)$ rule. If  we have  
$(p_{\maxN},m_{\maxN}) = (v,2)$ with $v \neq \minN $ in $A_{0}$ and $ A_{2}$,
then $t$ executes a $Reset$ rule and $v$ executes a $Write(\maxN)$ rule between $A_0$ and $A_2$.
\end{lemma}

\begin{proof}

Since $\maxN$ executes a $Write(\minN)$ rule in  $A_0 \mapsto A_1$ and $A_{2}\mapsto A_{3}$, then $r_{ts}=(Other, 2)$ 
in $A_{1}$ and $A_{3}$. By definition of $Write(\minN)$ rule, we also have $r_{ts}\neq (Other, 2)$ in $A_{2}$.   
Thus $\maxN$ must execute a $Write(\minN)$ rule   between $A_1$ and $A_2$ to modify $r_{ts}$. Let $ F   
\mapsto F'$  be the last transition  before $A_{2}$ in which $t$ executes such a rule.
Thus in $F$, either $(p_{t},m_{t}) = (v,m)$ with $m <2$, or $p_{t} \neq v$.  

First assume that $(p_{t},m_{t}) = (v,m)$ with $m <2$.   Thus, $m_{t} = 2$ in $A_{0}$ and $m_{t} < 2$ in $F$ .
There is only one way to decrease the value of an $m$ variable: to write 0 by executing a $Reset$  rule. On the second case,  $p_{t} \neq v$ in $F$ and $p_{t} = v$ in $A_{0}$. The only way to change the value of a non-null $p$ variable is that $t$ executes a $Reset$ rule between $A_{0}$ and $F$. Then in both cases, $t$ executes  a $Reset$  rule between $A_{0}$ and $A_2$.

Let $C_0 \mapsto C_1$ be the last $Reset$ executed by $t$ between $A_1$ and $A_2$ (we can have $C_0=A_1$). In $C_1$, $m_t=0$ and in $A_2$, $m_t=2$. Thus $t$ must execute an $Increase$ rule between $C_1$ and $A_2$ to set $m_t=2$.

Let $C_5 \mapsto C_6$ be the last $Increase$ rule executed by $t$ between $C_1$ and $A_2$. Since $t$ does not perform any $Reset$ from $C_6$ to $A_3$ by construction, then $p_t$ remains contant from $C_6$ to $A_3$. $p_t=v$ in $A_3$, thus $p_t=v$ in $C_5$ and so $t$ performs a $v$-$Increase$ rule in $C_5 \mapsto C_6$.  

First we assume that $v < t$. Let $C_2 \mapsto C_3$ be the first $Marriage(-)$ rule executed by $t$ between $C_1$ and $C_5$. Since $t$ does not perform any $Reset$ from $C_3$ to $C_{5}$ by construction, then $p_t$ remains contant from $C_3$ to $C_{5}$. $p_t=v$ in $C_5$, thus $p_t=v$ in $C_3$ and so $t$ performs a $Marriage(v)$ rule in $C_2 \mapsto C_3$.  Thus in $C_{2}$, $r_{vt}=(You,0)$, by the $Marriage(v)$ rule.  In $C_{5}$, $r_{vt}=(You,2)$, by the $v$-$Increase$ rule.  Thus,   $v$ has  performed at least once  $Write(t)$ between $C_2$ and $C_5$ and the lemma holds.

Second we assume that $t < v$. Let $C_2 \mapsto C_3$ be the first $Seduction(-)$ rule executed by $t$ between $C_1$ and $C_5$. Since $t$ does not perform any $Reset$ from $C_3$ to $C_{5}$ by construction, then $p_t$ remains contant from $C_3$ to $C_{5}$. $p_t=v$ in $C_5$, thus $p_t=v$ in $C_3$ and so $t$ performs a $Seduction(v)$ rule in $C_2 \mapsto C_3$.  Thus in $C_{2}$, $r_{vt}=(Idle,0)$, by the $Seduction(v)$ rule, and  in $C_{5}$, $r_{vt}=(You,1)$, by the $v$-$Increase$ rule.  Thus, $v$ has  performed at least once $Write(t)$ between $C_2$ and $C_5$ and this concludes the proof.
\end{proof}

\begin{lemma}\label{lemma:nb:seduction}
Let $({\minN},{\maxN})$ be an edge with ${\minN}<{\maxN}$. 
Let $\Exec$ be an execution containing three transitions $D_0 \mapsto D_1$, $D_{2}\mapsto D_{3}$ and $D_{4}\mapsto D_{5}$ with $D_{1} \mapsto^* D_{2}$ and $D_3\mapsto^* D_4$ and where $\minN$ executes a $Seduction(\maxN)$ rule. Then there exists a transition $D\mapsto D'$ between $D_2$ and $D_4$ where $\maxN$ executes a $Write(\minN)$ rule and with  in $D$: $p_t\neq null$ and $m_t=2$.
\end{lemma}

\begin{proof}
According to Lemma~\ref{lem:u:increase}.(2), $s$ executes a $Reset$ between $D_1$ and $D_2$. Let $C_0\mapsto C_1$ be the transition where $s$ does so for the first time. So $s$ executes a $t$-$Reset$ rule in $C_0\mapsto C_1$. And then according to Lemma~\ref{lem:lp:write}, $t$ executes a $Write(s)$ rule between $D_1$ and $C_0$. 

Observe now that the execution starting in configuration $D_2$ reaches all the assumptions made in Lemmas~\ref{petit1}. Indeed, before $D_2$, $t$ executes a $s$-rule and then $s$ executes a $Reset$ rule.

According to Lemma~\ref{lem:LP:entre2sed}, $s$ executes a $Reset$ rule between $D_3$ and $D_4$. Let $C_2\mapsto C_3$ be the  transition where $s$ does so for the first time. So $s$ executes a $t$-$Reset$ in $C_2\mapsto C_3$. In $C_2$, either $m_s>0$ or $m_s=0$. 

If $m_s>0$ in $C_2$ and since $m_s=0$ in $D_3$ by the $Seduction$ rule, then $s$ performs $m_s:=1$ between $D_3$ and $C_2$, let say in transition $B\mapsto B'$. In this transition, $s$ executes an $Increase$ rule. Recall that $p_s=t$ in $D_3$ since $s$ executes a $Seduction(t)$ rule, and that $s$ does not execute any $Reset$ between $D_3$ and $C_2$. Thus $p_s=t$ in $B$ and so $(p_s,m_s)=(t,1)$ in $B'$. By Lemma~\ref{petit1}, the edge $(s,t)$ is in a correct state in $B'$. Thus, by Corollary~\ref{coro:correctStateLP}, from this configuration, $s$ cannot execute any $Reset$, which contradict the fact that it does in $C_2\mapsto C_3$. Then $m_s=0$ in $C_2$. 

Since $(p_s,m_s)=(t,0)$ in $C_2$ with $s<t$  then, according to the $Reset$ rule, $r_{ts}.m=2$ in $C_2$. Moreover, since $r_{ts}.m=0$ in $D_2$ by the $Seduction(t)$ rule, then $t$ executes a $Write(s)$ rule between $D_2$ and $C_2$. This is the transition $D\mapsto D'$ and we have $m_t=2$ and $p_t\neq null$  in $D$.  This conclude the proof. 
\end{proof}

\subsection{Main Result}

Using Lemmas~\ref{lem:JC:entre2mar},~\ref{lem:u:BetweenTwoResets},~\ref{lem:u:3increase}, we can conclude on the time complexity of the algorithm.

\begin{theorem}\label{th:final}
The algorithm  stabilizes in  $O(n\Delta^3)$ moves. 
\end{theorem}

\begin{proof}
Let $(s,t)$ be an edge with $s<t$. By Theorem~\ref{theo:nb:seduction}, node $s$ can execute the $Seduction(t)$ rule $\mathcal{O}(\Delta)$ times. By Lemma~\ref{lem:JC:entre2mar}, between two executions of the $Marriage(s)$ rule by node $t$, node $s$ must execute a $Seduction(t)$ rule. This implies that $t$ can execute $\mathcal{O}(\Delta)$ $Marriage(s)$ rules. In total, a node can execute $\mathcal{O}(\Delta)$ $Seduction(-)$ and $Marriage(-)$ rules per neighbor, which gives a $\mathcal{O}(\Delta^2)$ total number of these rules, per node. Let now $u$ be a node. By Lemma~\ref{lem:u:BetweenTwoResets}, between two executions of the $Reset$ rule by node $u$, it must execute a $Marriage(-)$ or $Seduction(-)$ rule. Since we proved that it can do so $\mathcal{O}(\Delta^2)$ times, then it can execute the $Reset$ rule $\mathcal{O}(\Delta^2)$ times as well. Now by Lemma~\ref{lem:u:3increase}, between three executions of the $Increase$ rule, node $u$ must execute the $Reset$ rule. As a consequence, $u$ can execute the $Increase$ rule $\mathcal{O}(\Delta^2)$ times. Altogether, a node can execute at most $\mathcal{O}(\Delta^2)$ $Seduction(-)$, $Marriage(-)$, $Increase$ and $Reset$ rules. Let's call such rules high level rules. Each time it executes such a rule, a node may execute a $Write(-)$ rule $\mathcal{O}(\Delta)$ times to update all its registers. If it does not execute any high level rule, a node can execute at most $\mathcal{O}(\Delta)$ $Write$ rules. Thus in total a node can do $\mathcal{O}(\Delta^3)$ $Write$ rules. Finally, nodes can, in total, execute $\mathcal{O}(n\Delta^2)$ high level rules and $\mathcal{O}(n\Delta^3)$ $Write$ rules which gives a $\mathcal{O}(n\Delta^3)$ bound on the total number of moves.
\end{proof}

\section{Conclusion}
\label{sec:conclusion}
In this paper, we propose the first algorithm to solve the matching problem in link-register model. This algorithm is self-stabilizing, and takes at worst $O(n\Delta^3)$ moves before converging from the worst possible initialization, with the worst possible scheduling of communications. This is to be compared with similar solutions in state model, that converge in $O(m)$ moves. Indeed, asynchronous communications allow executions with a node taking steps in the algorithm ignoring the actual state of its neighbors. Moreover, to discard outdated values of the register, the matching process between two nodes requires a number of steps, to ensure that eventually, the two nodes agree regarding their marriage and will no longer take any move.

\end{document}